\documentclass[runningheads]{llncs}
\usepackage[T1]{fontenc}
\usepackage{graphicx}

\usepackage{mathtools,amsmath,amsfonts}
\usepackage{amssymb}
\usepackage{thm-restate}
\usepackage{stmaryrd}
\usepackage{xspace}
\usepackage{hyperref}

\usepackage{algorithm}
\usepackage[noend]{algpseudocode}
\usepackage{changepage}
\usepackage{paralist}
\usepackage{thm-restate}
\usepackage{multirow}
\usepackage[capitalise]{cleveref}
\usepackage{graphicx}
\usepackage{subcaption}
\usepackage{longtable}
\crefname{appendix}{}{}

\usepackage[
colorinlistoftodos, %
textwidth=\marginparwidth, %
textsize=scriptsize, %
]{todonotes}

\usepackage{mdframed}

\usepackage{pgf}
\usepackage{tikz}
\usetikzlibrary{shapes,shapes.geometric,fit,trees,decorations,arrows,automata,shadows,positioning,plotmarks,backgrounds,tikzmark}
\usetikzlibrary{calc,matrix,fit,petri,decorations.markings,decorations.pathmorphing,patterns,intersections,decorations.text}

\usepackage{graphicx}
\usepackage[utf8x]{inputenc}
\usepackage[american]{babel}
\usepackage[babel]{csquotes}

\usepackage{algorithm, algpseudocode}

\usepackage{amscd}
\usepackage{mathtools}
\usepackage{array}

\usepackage{xspace}
\usepackage{paralist}
\usepackage{xifthen}
\usepackage{url}

\usepackage{pgf}
\usepackage{tikz}

 \newcommand{\RomanNumeralCaps}[1]{\MakeUppercase{\romannumeral #1}}

\usepackage{expl3}[2012-07-08]
\ExplSyntaxOn
\cs_if_exist:NTF \fpeval
{ %
}
{ %
	\cs_new_eq:NN \fpeval \fp_eval:n
}
\ExplSyntaxOff

\colorlet{darkgreen}{green!80!black}
\colorlet{darkred}{red!80!black}
\usetikzlibrary{arrows, automata, shapes}
\tikzset{auto, >= stealth}
\tikzset{every edge/.append style={thick, shorten >= 1pt}}
\tikzset{initial/.style={draw, thick, <-, shorten <=1pt}}
\tikzset{player0/.style = {draw, thick, shape=circle, minimum size=5mm}}
\tikzset{player1/.style = {draw, thick, shape=rectangle, minimum size=5mm}}
\tikzset{playerR/.style = {draw, thick, shape=diamond, minimum size=5mm}}
\tikzset{bplayer0/.style = {draw, thick, shape=ellipse, minimum size=5mm,text width=1.1cm}}
\tikzset{bplayer1/.style = {draw, thick, shape=rectangle, minimum size=5mm,text width=1.6cm}}

\usetikzlibrary{shapes.geometric, arrows}

\tikzstyle{startstop} = [rectangle, rounded corners, 
minimum width=0.1cm, 
minimum height=0.5cm,
text centered, 
draw=black, 
]

\tikzstyle{process} = [trapezium, 
trapezium stretches=true, %
trapezium left angle=70, 
trapezium right angle=110, 
minimum width=0.1cm, 
minimum height=0.5cm, 
draw=black, 
]

\tikzstyle{io} = [rectangle, 
minimum width=0.1cm, 
minimum height=0.5cm, 
text centered, 
draw=black, 
]

\tikzstyle{decision} = [diamond, 
aspect=3,
text centered, 
draw=black, 
]
\tikzstyle{arrow} = [thick,->,>=stealth]

\newcommand{\bigO}{\mathcal{O}}

\newcommand{\abs}[1]{\left\lvert #1 \right\rvert}

\newcommand{\LTLeventually}{\lozenge}
\newcommand{\LTLalways}{\square}

\newcommand{\lang}{\mathcal{L}}

\newcommand{\game}{\mathcal{G}}
\newcommand{\gamegraph}{\ensuremath{G}}
\newcommand{\play}{\ensuremath{\rho}}
\newcommand{\plays}{\ensuremath{\textsf{plays}}}

\newcommand{\spec}{\ensuremath{\Phi}}

\newcommand{\vertex}{V}

\newcommand{\p}[1]{\ensuremath{\text{Player}~#1}}

\newcommand{\win}{\mathcal{W}}

\newcommand{\wino}{\win_1}

\newcommand{\wini}{\win_i}

\newcommand{\SafeAlgo}[1]{\textsc{AlmostSafe}(#1)}

\newcommand{\minsupp}{\ensuremath{\textsf{minProb}}}
\newcommand{\minsupport}[2]{\ensuremath{\minsupp_{#1}(#2)}}

\newcommand{\strat}{\ensuremath{\pi}}
\newcommand{\stratI}[1]{\ensuremath{\strat_{#1}}}

\newcommand{\strato}{\ensuremath{\stratI{1}}}
\newcommand{\stratt}{\ensuremath{\stratI{2}}}
\newcommand{\strati}{\ensuremath{\stratI{i}}}

\newcommand{\Strat}{\ensuremath{\Pi}}
\newcommand{\StratI}[1]{\ensuremath{\Strat_{#1}}}

\newcommand{\distro}{d_1}
\newcommand{\distrt}{d_2}
\newcommand{\distri}{d_i}

\newcommand{\transition}{\ensuremath{\delta}}
\newcommand{\probdistributions}{\ensuremath{\mathcal{D}}}
\newcommand{\support}{\ensuremath{\textsf{supp}}}

\newcommand{\probability}{\Pr}
\newcommand{\orprob}[4]{\textstyle \Pr_{#1}^{#2,#3}(#4)}

\newcommand{\buchi}{\ifmmode B\ddot{u}chi \else B\"uchi \fi}
\newcommand{\cobuchi}{\ifmmode co\text{-}B\ddot{u}chi \else co-B\"uchi \fi}

\newcommand{\safeTemp}{\textsc{SafetyTemp}}

\newcommand{\liveTemp}{\textsc{LivenTemp}}

\newcommand{\BuchiAlgo}[1]{\textsc{Almost\buchi}(#1)}
\newcommand{\CobuchiAlgo}[1]{\textsc{AlmostCo-\buchi}(#1)}

\newcommand{\buchiTemp}{\textsc{B\"uchiTemp}}
\newcommand{\cobuchiTemp}{\textsc{coB\"uchiTemp}}

\newcommand{\livePartitions}{\textsf{P}}
\newcommand{\livePartition}{U}

\newcommand{\infPlay}[1]{\ensuremath{\textsf{inf}}(#1)}
\newcommand{\infStrat}[1]{\ensuremath{\textsf{inf}}(#1)}

\newcommand{\template}{\ensuremath{\Lambda}}

\newcommand{\funcSafe}{\ensuremath{\textsf{S}}}
\newcommand{\funcLive}{\ensuremath{\textsf{H}}}
\newcommand{\funcCoLive}{\ensuremath{\textsf{C}}}
\newcommand{\targetSet}{\ensuremath{I}}

\newcommand{\templatesafe}{\ensuremath{\template_{\textsc{unsafe}}}}
\newcommand{\templatelive}{\ensuremath{\template_{\textsc{live}}}}

\newcommand{\templategrlive}{\ensuremath{\template_{\textsc{live}}}}

\newcommand{\templatecolive}{\ensuremath{\template_{\textsc{colive}}}}

\newcommand{\pre}[2]{\textsf{pre}\ensuremath{_{#1}(#2)}}

\makeatletter
\newsavebox{\@brx}
\newcommand{\llangle}[1][]{\savebox{\@brx}{\(\m@th{#1\langle}\)}%
	\mathopen{\copy\@brx\kern-0.5\wd\@brx\usebox{\@brx}}}
\newcommand{\rrangle}[1][]{\savebox{\@brx}{\(\m@th{#1\rangle}\)}%
	\mathclose{\copy\@brx\kern-0.5\wd\@brx\usebox{\@brx}}}

\makeatother

\makeatletter 
\newif\ifFIRST
\newif\ifSECOND
\let\LISTOP\relax
\newcommand{\List}[4][\;]{#3#1%
	\FIRSTtrue
	\@for\i:=#2\do{%
		\ifFIRST\LISTOP{\i}\FIRSTfalse\else,\LISTOP{\i}\fi%
	}%
	#1#4%
	\let\LISTOP\relax
}
\makeatother

\makeatletter

\newcommand{\VSet}[2][]{\let\LISTOP\val\List[#1]{#2}{\{}{\}}}

\newcommand{\VTuple}[2][]{\let\LISTOP\val\List[#1]{#2}{(}{)}}

\newcommand{\subaction}{\ensuremath{\gamma}}
\newcommand{\subactiono}{\ensuremath{\subaction_{1}}}
\newcommand{\subactiont}{\ensuremath{\subaction_{2}}}

\newcommand{\acto}{\ensuremath{a}}
\newcommand{\actt}{\ensuremath{b}}

\newcommand{\safeactions}[3]{\ensuremath{\textsf{A}_{#1}^{#2}(#3)}}
\newcommand{\forwardingactions}[3]{\ensuremath{\textsf{B}_{#1}^{#2}(#3)}}
\newcommand{\actions}{\ensuremath{\Gamma}}

\newcommand{\actiono}{\actions_1}
\newcommand{\actiont}{\actions_2}

\newcommand{\Apre}[3]{\textsf{Apre}\ensuremath{_{#1}(#2,#3)}}

\newcommand{\AFpre}[4]{\textsf{AFpre}\ensuremath{_{#1}(#2,#3,#4)}}

\newcommand{\alphabets}{\Sigma}

\newcommand{\toolname}{\texttt{ConSTel}\xspace}

\definecolor{colorblindgreen}{HTML}{004D40}
\definecolor{colorblindblue}{HTML}{1E88E5}
\definecolor{colorblindorange}{HTML}{FFC107}
\definecolor{colorblindred}{HTML}{D81B60}
\definecolor{colorblindbluealt}{HTML}{00BFFF}
\definecolor{colorblindorangealt}{HTML}{FF7F00}

\newcommand{\robotC}{\ensuremath{{R}_C}\xspace}
\newcommand{\robotE}{\ensuremath{{R}_E}\xspace}
\newcommand{\clockwise}{\ensuremath{\circlearrowright}}
\newcommand{\anticlockwise}{\ensuremath{\circlearrowleft}}

\graphicspath{{./img/}{./img/svg/}}

\begin{document}
\title{Concurrent Permissive Strategy Templates\thanks{Authors are ordered alphabetically. All authors are supported by the DFG project 389792660 TRR 248-CPEC. 
Additionally A.-K.~Schmuck is supported by the DFG project SCHM 3541/1-1. C.~Baier, C.~Chau, and S.~Kl\"uppelholz are supported by the German Federal Ministry of Research, Technology and Space within the project SEMECO Q1 (03ZU1210AG) and funded by the German Research Foundation (DFG, Deutsche Forschungsgemeinschaft) as part of Germany's Excellence Strategy - EXC 2050/2 - Project ID 390696704 - Cluster of Excellence ``Centre for Tactile Internet with Human-in-the-Loop'' (CeTI) of Technische Universit\"at Dresden.}}

\author{
Ashwani Anand\inst{1}
\and
Christel Baier\inst{2}
\and
Calvin Chau\inst{2}
\and
Sascha Klüppelholz\inst{2}\and\\
Ali Mirzaei\thanks{This research was conducted while Ali Mirzaei was interning at MPI-SWS, Germany (\email{ali.mirzaei78@sharif.edu}).}
 \and
Satya Prakash Nayak\inst{1}
\and
Anne-Kathrin Schmuck\inst{1}
}
\institute{$^1$ Max Planck Institute for Software Systems, Kaiserslautern, Germany\\
\email{\{ashwani,sanayak,akschmuck\}@mpi-sws.org}\\
$^2$ Technische Universität Dresden, Germany\\
\email{\{christel.baier,calvin.chau,sasha.klueppelholz\}@tu-dresden.de}
}
\maketitle              %
\begin{abstract}
Two-player games on finite graphs provide a rigorous foundation for modeling the strategic interaction between reactive systems and their environment. While concurrent game semantics naturally capture the synchronous interactions characteristic of many cyber-physical systems (CPS), their adoption in CPS design remains limited. Building on the concept of permissive strategy templates (PeSTels) for turn-based games, we introduce concurrent (permissive) strategy templates (ConSTels) -- a novel representation for sets of randomized winning strategies in concurrent games with Safety, Büchi, and Co-Büchi objectives. ConSTels compactly encode infinite families of strategies, thereby supporting both offline and online adaptation. Offline, we exploit compositionality to enable incremental synthesis: combining ConSTels for simpler objectives into non-conflicting templates for more complex combined objectives. Online, we demonstrate how ConSTels facilitate runtime adaptation, adjusting action probabilities in response to observed opponent behavior to optimize performance while preserving correctness. We implemented ConSTel synthesis and adaptation in a prototype tool and experimentally show its potential. 

\end{abstract}

	\section{Introduction}
	Two player games on finite graphs \cite{gamesbook} provide a powerful abstraction for modeling the strategic interactions between reactive systems and their environment. 
As such, they have been successfully applied to both software \cite{HenzingerBook,finkbeiner2016synthesis} and cyber-physical system (or component) design \cite{TabuadaBook,kress2018synthesisForRobotsReview,YIN2024100940} to ensure that interactions adhere to stringed correctness requirements. Such requirements are typically temporal -- modeling the correct interaction of processes over time -- and can be formalized in temporal logic or, more generally, as an $\omega$-regular language. %

Depending on the mode of interaction between systems (resp. components) the interaction semantics of players in the resulting game differ. Asynchronous interactions result in \emph{turn-based} games \cite{pnueli1989synthesis,EmersonJutla91}, where in each round only one of the two players can choose among several moves. On the other hand, synchronous interactions lead to \emph{concurrent} games \cite{AlfaroHK98,AlfaroHenzinger_Concurrent2000,Chatterjee07}, where
in each round both players can choose simultaneously and independently among several moves. 

In fact, many strategic interactions of reactive systems are naturally synchronous -- especially in the context of cyber-physical system design. Multiple robots are moving into different regions of the states space concurrently and do not wait for each other to complete a motion. Similarly, human-robot cooperation in unstructured environments, e.g.\ emptying a dishwasher, are only truly helpful to the human if both work alongside (synchronously) and are not forced to take turns (as in structured environments like manufacturing lines). Finally, physical conditions are also typically changing synchronously to the system's actions, e.g. wind conditions might become challenging \emph{while} a drone is landing -- not `waiting' for the drone to land before causing turbulence. 

Despite these natural synchronous strategic interactions of technological systems with their environment, almost all existing work -- especially in the context of strategic control for CPS -- focuses on the use of turn-based game semantics to abstract system interaction (see  \cite{kress2018synthesisForRobotsReview,YIN2024100940} for overviews). %
While traditional approaches typically use games-based synthesis as a black-box in CPS design, a recent line of work introduced \emph{permissive strategy templates} ~(PeSTels)~\cite{AnandMNS23,AnandNS_PeSTels2023,AnandNS24} as a new concise data structure which captures an infinite number of winning strategies in turn-based graph games. This \emph{permissiveness} allows for the adaptation of strategies both \emph{offline}, e.g., for incremental or distributed synthesis \cite{AnandNS_PeSTels2023,ASN_HSCC24,NS_TACAS24} and \emph{online}, e.g., to handle actuation faults or enable dynamic multi-layered control \cite{NNS_HSCC24} of CPS. While the work in \cite{PhalakarnPH_ICTAC2024} extends templates to almost sure winning in stochastic parity games, it does not consider the concurrent setting. %

\emph{Concurrent games} on the other hand, were first considered by Shapley \cite{shapley1953stochastic} and have been studied since~\cite{EtessamiY10,AlfaroHK98,AlfaroHenzinger_Concurrent2000,Chatterjee07}. For a detailed historical account, see \cite{gamesbook}. In particular, the work of \cite{AlfaroHenzinger_Concurrent2000} provides algorithms for solving concurrent $\omega$-regular games. It has been shown that -- in contrast to turn-based games -- concurrent games might require players to use \emph{randomization} over their strategic choices to win. This makes the corresponding solution algorithms strictly more complex than existing algorithms for turn-based games. 

While concurrent games are theoretically well-studied, tool support for $\omega$-regular objectives is scarce. The \textsc{Mocha} tool \cite{AlurHMQRT98} considers concurrent systems in the form of reactive modules \cite{AlurH96} and supports Alternating Temporal Logic (ATL), but appears to be no longer available. The \textsc{Gavs+} platform \cite{ChengKLB11} only handles concurrent reachability games. \textsc{Prism}-games \cite{Kwiatkowska2020} support concurrent stochastic games, but not Büchi or Co-Büchi conditions. \textsc{Praline} \cite{Brenguier13} can check the existence of \emph{pure} (i.e. no randomization) strategy Nash equilibria. \textsc{Mcmas} \cite{LomuscioQR17} and \textsc{Eve} \cite{GutierrezNPW18} also check pure Nash equilibria.
This limited support for richer property classes, such as $\omega$-regular objectives, additionally hinders the adaptation of concurrent games as abstractions for CPS design, despite their natural semantic fit.

\smallskip
\noindent\textbf{Contribution.}
Based on the success of PeSTels for CPS design and the observation that most of their interactions are synchronous, this paper develops \emph{concurrent strategy templates} (ConSTels) for  Safety, Büchi and Co-Büchi objectives. We further present a prototype tool \toolname which automatically computes ConSTels for these games. 
We exploit the advantages of ConSTels over (classical) winning strategies in concurrent games w.r.t.\ both the (i) \emph{offline} and (ii) \emph{online} adaptation of strategies, respectively.

(i) \emph{Offline adaptation:} We formalize incremental synthesis, i.e., the combination of ConSTels computed for individual Büchi or Co-Büchi objectives into a new ConSTel for the combined objective, which are strictly more complex than both individually. Our techniques allow to compute winning strategies for these multi-dimensional objectives if the resulting ConSTels are non-conflicting. While this approach is necessarily incomplete, we show experimentally, that it is successful on several examples adapted from the SYNTCOMP benchmark suite \cite{JacobsPABCCDDDFFKKLMMPR24} (see \cref{figure:lineplot-conflict}). This is particularly interesting, as there is currently no other tool support for concurrent parity games.

(ii) \emph{Online adaptation:} We experimentally showcase the potential of our templates for the purpose of runtime optimization. As ConSTels capture an infinite set of randomized strategies, we can adapt the probability distributions over actions for a certain state during runtime. Given an opponent with a fixed but unknown action distribution, ConSTels allow to infer this distribution at runtime and adapt strategies accordingly. This enables efficient satisfaction of objectives (see \cref{fig:robot_exp}). While this experiment solely serves as a proof-of-concept, we note that it shows the potential of ConSTels to integrate matrix games~\cite{NRTV2007} per state to resolve local dependencies of action probabilities among agents within a strategic, long-term game. We leave this exploration for future work. 

While there exists limited work on concurrent games in robotic planning \cite{WangDCK16}, multi-agent learning \cite{BowlingV03} and collective strategy synthesis in multi-agent systems \cite{XiongL16} via ATL \cite{AlurHK02,AlurHMQRT98}, to the best of our knowledge, we present the first approach to permissive strategies for concurrent games motivated by CPS applications. %

	\section{Preliminaries}
	This section recalls all necessary preliminaries on concurrent games from \cite{AlfaroHenzinger_Concurrent2000} and strategy templates from \cite{AnandNS_PeSTels2023}.

\noindent\textbf{Notation.} 
For a finite alphabet $\alphabets$, we use $\alphabets^*$, $\alphabets^+$, and $\alphabets^\omega$ to denote the set of all finite, non-empty finite, and infinite words over $\alphabets$, respectively.
We define $\alphabets^\infty := \alphabets^* \cup \alphabets^\omega$ and for any word $\omega \in \alphabets^\infty$, $\omega_i$ denotes the $i$-th symbol in $\omega$ and $\omega_{\leq i}:=\omega_1\omega_2\hdots\omega_i$. %
A \emph{probability distribution} on a finite set $A$ is a function $p : A \rightarrow [0,1]$ s.t.\ $\sum_{a \in A} p(a) = 1$. The set of all probability distributions on $A$ is denoted by $\probdistributions(A)$. The \emph{support} of a probability distribution $p$ on set $A$ is defined by $\support(p):=\{a \in A \mid p(a) > 0\}$.
For $B \subseteq A$, we define $p(B): = \sum_{b \in B} p(b)$. %

\smallskip
\noindent\textbf{Concurrent Game Graphs.} A \emph{two-player concurrent game graph} is a tuple $G=(V,\actions_1,\actions_2,\transition)$ where $V$ is a finite set of states,
$\actions_1$ and $\actions_2$ are finite sets of actions for $\p{1}$ and $\p{2}$, respectively, and $\transition:V \times \actions_1 \times \actions_2 \rightarrow V$ is a transition function.
At every state $v \in \vertex$, we denote all possible actions for player $i$ by $\actions_i(v) \subseteq \actions_i$ and, $\p{1}$ chooses an action $a \in \actions_1(v)$ and simultaneously and independently $\p{2}$ chooses an action $b \in \actions_2(v)$. Then the game proceeds to a state $v' \in V$ such that $\transition(v,a,b)=v'$.
A \emph{play} from a state $v_0$ is a finite or infinite sequence of states $\play\in V^\infty$ such that $\play_0=v_0$  and for every $0\leq i \leq |\play|$, there are actions $a_i \in \actions_1(\play_i)$ and $b_i \in \actions_2(\play_i)$, such that $\transition(\play_i,a_i,b_i)=\play_{i+1}$.
We write $\infPlay{\play}$ to denote the set of all states that appear infinitely often in $\play$.

\smallskip
\noindent\textbf{Strategies.} %
A \emph{(randomized) strategy} for player $i \in \{1,2\}$ over a concurrent game graph $\gamegraph$ is a function $\strati : V^+ \rightarrow \probdistributions(\actions_i)$ such that for all $\play v \in \vertex^+$, it holds that $\support(\strati(\play v)) \subseteq \actions_i(v)$. We collect all player $i$ strategies over $\gamegraph$ in the set $\Strat_i$.
A strategy $\strat$ is \emph{deterministic} (or pure) if for every \emph{play} $\play$, the distribution $\strat(\play)$ assigns probability $1$ to a single action, and \emph{randomized} otherwise. 
A strategy $\strat$ is \emph{memoryless} if for all $\play v, \play' v \in V^+$, it holds that $\strat(\play v) = \strat(\play' v)$. 
An \emph{infinite} play $\play\in V^\omega$ is \emph{compliant} with a $\p{1}$ strategy $\stratI{1}$ if for every $i\in\mathbb{N}$ there exist actions $\acto \in \actions_1(\play_i)$ and $\actt \in \actions_2(\play_i)$ s.t. $\transition(\play_i,\acto,\actt)=\play_{i+1}$ and $\acto \in \support(\strato(\play_{\leq i})) $, i.e. $\strato(\play_{\leq i})(\acto) > 0$.
A strategy for $\p{2}$ is defined analogously. 
We refer to a play compliant with $\strati$ and a play compliant with both $\strato$ and $\stratt$ as a \emph{$ \strati $-play} and a \emph{$ \strato\stratt $-play}, respectively and collect all $ \strati $-plays in the set $\plays(\strat)$.

\smallskip
\noindent\textbf{Winning Conditions.}
We consider safety, Büchi and co-Büchi winning conditions over \emph{infinite} plays. Given a set of safe states $I\subseteq V$, an infinite play $\play\in V^\omega$ over $\gamegraph$ is winning in the \emph{concurrent safety game} $(\gamegraph,\Box I)$ iff for all $i\in\mathbb{N}$ holds $\play(i)\in I$. Given a set of Büchi states $T\subseteq V$, an infinite play $\play\in V^\omega$ over $\gamegraph$ is winning in the \emph{concurrent Büchi game} $(\gamegraph,\Box\Diamond T)$ iff for all $i\in\mathbb{N}$ exists $j\geq i$ s.t.\ $\play(j)\in T$.  Given a set of co-Büchi states $\overline{T}\subset V$ with $T = V\setminus \overline{T}$, an infinite play $\play\in V^\omega$ over $\gamegraph$ is winning in the \emph{concurrent co-Büchi game} $(\gamegraph,\Diamond\Box T)$ iff there exists $i\in\mathbb{N}$ s.t.\ for all $j\geq i$ holds $\play(j)\notin \overline{T}$. We write $(\gamegraph, \spec)$ with $\spec\in\{\Box I, \Box \Diamond I, \Diamond\Box I\}$ s.t.\ $I\subseteq V$ to denote any concurrent game and
denote by $\mathcal{L}(\spec)$ the corresponding set of winning plays.

\smallskip
\noindent\textbf{Almost-Sure Winning.} 
We consider almost-sure winning in concurrent games. 
A strategy $\strat$ for $\p{1}$ is \emph{almost-sure winning}  (winning for short) from $v$ in $(\gamegraph, \spec)$ if $\probability_v^\strat(\lang(\spec)) = 1$, where $\probability_v^{\strati}(X)$ denotes the probability that a $\strati$-play starting from $v$ belongs to $X$.
The \emph{almost sure winning region} (winning region for short) for $\p{1}$ in the concurrent game $(\gamegraph, \spec)$ is the set $\wini(\spec)\subseteq V$ of states from which $\p{1}$ has a strategy to win.
A strategy $\strat$ is a \emph{winning strategy} for $\p{1}$ in $(\gamegraph, \spec)$, if it is winning from every $v\in \wini(\spec)$.

\smallskip
\noindent\textbf{Strategy Templates.}
Formally, a strategy template $\template$ is a (possibly infinite) set of $\p{1}$ strategies, i.e.\ $\template\subseteq\StratI{1}$. Then, a $\p{1}$ strategy $\strat_1$ is said to \emph{follow} the template $\template$, if $\strat_1 \in \template$.
We say a strategy template $\template$ is \emph{winning from a state $v$} in a game $(\gamegraph,\spec)$, if every $\p{1}$ strategy following the template $\template$ is winning in $(\gamegraph,\spec)$ from $v$. Moreover, we say a strategy template $\template$ is \emph{winning}, if it is winning from every state in $\wino$. In addition, we call $\template$ \emph{maximally permissive} for $\game$, if every $\p{1}$ strategy $\pi$ which is winning in $\game$ also follows $\template$.

	\section{Overview and Problem Statement}\label{sec:problem}
    bAs formalized in the previous section, winning strategy templates directly generalize from turn-based to concurrent games, due to their simple definition as sets of strategies. Their main advantage over classical winning strategies, however, is rooted in a simple and local data structure which concisely describes these strategy sets. Concretely, for turn-based (parity) games, it is shown in \cite{AnandNS_PeSTels2023} that winning strategy templates can be defined as conjunctions of three simple types of \p{1} edges conditions, namely 
\begin{inparaenum}[(i)]
    \item \emph{unsafe edges}: edges that \p{1} should never take,
    \item \emph{co-live edges}: edges that may only be taken by \p{1} finitely many times along a play, and
    \item \emph{live-groups}: sets of edges such that, if the source state of some live edges is visited infinitely often, \p{1} must take at least one edge from the set infinitely often.
\end{inparaenum}
In particular, these conditions can be extracted by careful bookkeeping during the solution (i.e., the winning region computation) of a safety, Büchi and co-Büchi game, respectively. %

The main contribution of this paper is the automatic extraction of a similarly concise data structure to capture a permissive set of winning strategies in concurrent safety, Büchi and co-Büchi games. Due to the synchronous game semantics, the need for randomized strategies, and the substantially different solution algorithms for concurrent games, the formalization and the computation of these edge conditions is substantially different from the turn-based case. %

We use an example to outline the general construction of concurrent strategy templates which allows us to illustrate their advantages over classical strategies. In addition, we use this example to outline the remaining structure of this paper which formalizes all concepts previewed in this example.

\begin{figure}[t]
  \begin{subfigure}[b]{0.45\textwidth}
	\centering
    \resizebox{0.9\linewidth}{!}{
     \begin{minipage}[b]{\textwidth}
         \centering
        \begin{tikzpicture}[scale=1.8]
		\fill[green!50!black!60] (1,1) rectangle (2,2); %
		\fill[green!50!black!60] (0,0) rectangle (1,1); %
		
		\draw[line width=2pt] (0,0) rectangle (2,2);
		\draw[line width=2pt] (1,0) -- (1,2);
		\draw[line width=2pt] (0,1) -- (2,1);

		\node at (0.5,1.5) {};
		\node at (1.5,1.5) {};
		\node at (0.5,0.5) {};
		\node at (1.5,0.5) {};

		\node[draw, circle, fill=blue!30, inner sep=1pt] at (0.5,1.5) {\robotC};

		\node[draw, circle, fill=red!30, inner sep=1pt] at (1.5,0.5) {\robotE};
	\end{tikzpicture}
	
    \caption{Grid World} \label{fig:grid_map}   
 \end{minipage}
}
 \end{subfigure}
      \begin{subfigure}[b]{0.45\textwidth}
      \centering
     \resizebox{0.9\linewidth}{!}{
     \begin{minipage}[b]{\textwidth}
\centering     
    \begin{tikzpicture}[->,>=stealth',shorten >=1pt,auto,node distance=2.7cm,semithick]
        \node[state,initial,line width=2pt] (I) {$S_0$};
        \node[state,right of=I,line width=2pt] (1) {$S_1$};
        \node[state,below of=1,line width=2pt] (2) {$S_2$};
        \node[state,below of=I,line width=2pt] (E) {$S_e$};

        \path (1) edge [bend left] node[sloped,  pos=0.25] {\small $(\clockwise,\anticlockwise)$} (I);
        \path (1) edge [bend left] node[sloped, pos=0.75] {\small $(\anticlockwise,\clockwise)$} (I);
        
        \path (I) edge [bend left] node[sloped, above, pos=0.75] {\small $(\anticlockwise,\clockwise)$} (1);
        \path (I) edge [bend left] node[sloped, above, pos=0.25] {\small $(\clockwise,\anticlockwise)$} (1);

        \path (I) edge [bend left] node[sloped, below, pos = 0.25] {\small $(\clockwise,\clockwise)$} 
        node[sloped,below, pos = 0.75,rotate=180] {\small $(\anticlockwise,\anticlockwise)$}
        (E);
        \path (E) edge [bend left] node[sloped, below, pos=0.25] {\small $(\clockwise,\clockwise)$} (I);
        \path (E) edge [bend left] node[sloped, below, pos = 0.75,rotate=180] {\small $(\anticlockwise,\anticlockwise)$} (I);

        \path (2) edge [bend left] node[sloped, below, pos=0.25] {\small $(\clockwise,\clockwise)$} (1);
        \path (2) edge [bend left] node[sloped, below, pos=0.75,rotate=180] {\small $(\anticlockwise,\anticlockwise)$} (1);
        \path (1) edge [bend left] node[sloped, below, pos=0.75,rotate=180] {\small $(\anticlockwise,\anticlockwise)$} (2);
        \path (1) edge [bend left] node[sloped, below, pos=0.25] {\small $(\clockwise,\clockwise)$} (2);

        \path (2) edge [bend left] node[sloped, below, pos=0.25] {\small$(\anticlockwise,\clockwise)$} (E);
        \path (2) edge [bend left] node[sloped, below, pos=0.75] {\small$(\clockwise,\anticlockwise)$} (E);
        
        \path (E) edge [bend left] node[sloped, below, pos=0.25] {\small$(\anticlockwise,\clockwise)$} (2);
        \path (E) edge [bend left] node[sloped, below, pos=0.75] {\small$(\clockwise,\anticlockwise)$} (2);
        
        \path (E) edge[loop left=45] node[above left] {$(N,\clockwise)$} (E);
        \path (E) edge[loop left=45] node[below left] {$(N,\anticlockwise)$} (E);
    \end{tikzpicture}
    
    \caption{Concurrent Game Graph}   \label{fig:robot_game} 
\end{minipage}
}
 \end{subfigure}

    \caption{Motivating Example. Left: Two robots, $\robotC$ and $\robotE$, moving in a grid world. $\robotC$ must enter a green region \emph{alone} to fulfill a task. Both robots can move either clockwise ($\clockwise$) or anti-clockwise ($\anticlockwise$). 
    Right: Concurrent game graph capturing the synchronous strategic interaction of $\robotC$ and $\robotE$.} \label{fig:robot_all} 
    \end{figure}
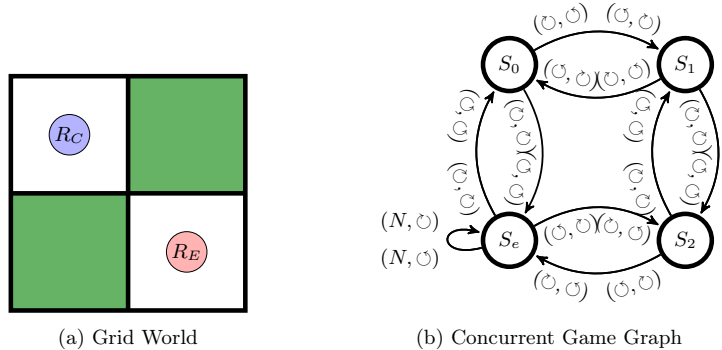

\smallskip
\noindent\textbf{Modeling Strategic Interaction as a Concurrent Game.}
We consider two robots -- a controlled robot $\robotC$ and an environment robot $\robotE$ -- in the very simple grid world depicted in \cref{fig:grid_map} where $\robotC$ must enter a green region \emph{alone} in order to fulfill a particular task. Both robots can move either clockwise or anti-clockwise through the grid. %

This strategic interaction can be modeled as a concurrent game between $\robotC$ as $\p{1}$ and $\robotE$ as $\p{2}$ over the game graph $\gamegraph$ depicted in \cref{fig:robot_game}. The game has four states, distinguishing all relevant configurations of robot positions in the workspace. The initial state $S_0$ represents the configurations where both robots are in diagonally opposite white corners of the grid.
The state $S_1$ (resp. $S_2$) represents the configurations, where both robots are in a green cell (resp. a white cell) together. Finally, the escape state $S_e$ represents the configurations where $\robotC$ is in a green cell and $\robotE$ is in a white cell.
In every state, both players can choose to move clockwise ($\clockwise$) or anti-clockwise ($\anticlockwise$) to the adjacent cell.
To simplify the resulting game, we assume for illustration purposes, that within the escape state $S_e$, $\p{1}$ (the robot $\robotC$) can choose to stay ($N$) in the cell and then, irrespective of the $\p{2}$'s action, the robots remain in $S_e$.

\smallskip
\noindent\textbf{Safety Templates.}
Let us first consider a very simple safety game where both robots start in $S_e$ and $\robotC$ must ensure that the play stays in $S_e$. Here, the winning region is $\wino(\Box S_e)=\{S_e\}$ and $\robotC$ must never choose actions from the set $\{\clockwise, \anticlockwise\}$ to win from states in $\wino$. This is modeled by a safety template which allows all strategies which do not assign positive probability to actions specified by the safety function $\funcSafe(S_e)=\{\clockwise, \anticlockwise\}$, see \cref{sec:safety} for details.

\smallskip
\noindent\textbf{Live-Group Templates.}
Let us now consider the \buchi game $(\gamegraph,\LTLalways \LTLeventually \{S_e\})$. 
Note that, in this game, the robot $\robotC$ can not win from the initial configuration $S_0$ by choosing moves deterministically in every state. For every deterministic choice of $\robotC$ in states $\{S_0, S_2\}$, $\robotE$ can choose an action which prevents reaching $S_e$.
Therefore, $\robotC$ must choose both actions with positive probability in states $\{S_0, S_2\}$ to ensure that, irrespective of $\robotE$'s strategy, the game will reach $S_e$ with probability $1$. In this case, however, we have $\wino(\Box \Diamond S_e)=V$ and a winning strategy template for $\robotC$ in this game would ask to assign positive probability to both moves $\{\clockwise, \anticlockwise\}$ in every state $\{S_0, S_2\}$, while in states $\{S_e, S_1\}$, the strategy template would not restrict $\robotC$'s moves to remain permissive.

Formally, we encode this strategy template as a live-group template $\template_1 = \templategrlive(\funcLive,\livePartitions)$, where $\funcLive$ is a function that maps every state to a set of action sets such that each action set needs to be assigned positive probability, and $\livePartitions$ allows to have multiple liveness constraints from different states to ensure permissiveness.
In this example, with the above intuition, we have $\funcLive(S_0) = \funcLive(S_2) = \{\{\clockwise\}, \{\anticlockwise\}\}$, and $\funcLive(S_e) = \{\actiono(S_e)\}$, $\funcLive(S_1) = \{\actiono(S_1)\}$.
However, we actually do not need to satisfy the liveness constraints from both $S_1$ \emph{and} $S_2$, but only from one of them (along with the liveness constraints of state $S_1$) to guarantee reaching state $S_e$. This is captured by combining $S_1$ and $S_2$ into the same \emph{live group}, leading to $\livePartitions = \{\{S_0, S_2\}, \{S_1\}\}$. The resulting live group template only requires a following strategy to assign positive probability to either $\funcLive(S_1)$ or $\funcLive(S_2)$ to be a winning strategy for $\robotC$. In the general setting where $\wino(\Box \Diamond S_e)\subsetneq V$, we additionally need a safety template to prevent strategies to leave $\wino$, as a live-group template for Büchi games is only defined over winning states. We discuss live-group templates formally in \cref{sec:buechi}.

\smallskip
\noindent\textbf{Co-Live Templates.}
Now, let us consider a \cobuchi game $(\gamegraph,\LTLeventually \LTLalways \{S_1, S_2, S_e\})$. We observe that $\robotC$ cannot ensure not visiting the initial state $S_0$ from states $\{S_1, S_2\}$. %
Hence, the only way for $\robotC$ to win this game is to ensure that (i) it always eventually reaches $S_e$ and (ii) from some point onward, it always stays in the escape state $S_e$. A permissive template is therefore composed of two parts. To realize (i) it requires the live-group template $\templategrlive(\funcLive,\livePartitions)$ computed for the Büchi game $(\gamegraph,\LTLalways \LTLeventually \{S_e\})$. To realize (ii), within state $S_e$, from some point onwards, a winning strategy must assign probability $0$ to both moves $\{\clockwise, \anticlockwise\}$. This can be encoded as a 
co-liveness template via a co-live function which maps a state to all actions whose accumulated probability must be bounded over time, so that it is eventually not taken again. In this example, we have $\funcCoLive(S_e) = \{\{\clockwise\}, \{\anticlockwise\}\}$, and for every other state $s$, $\funcCoLive(s) = \emptyset$. Again, if the winning region does not span the entire states pace, an additional safety template is required. In addition, co-Büchi games might require a more complex combination of co-live and live-group templates in the general case, leading to a more involved extraction algorithm. We discuss co-live templates formally in \cref{sec:cobuechi}.

\smallskip
\noindent\textbf{Offline Adaptability.}
One major advantage of the permissiveness of ConSTels is their induced offline adaptability. As an example, consider a game with the conjunction of the above discussed Büchi and co-Büchi winning conditions, i.e., $\spec = \LTLalways \LTLeventually \{S_e\} \land \LTLeventually \LTLalways \{S_1, S_2, S_e\}$. 
In this case, we can compose the constraints of the two strategy templates by taking their intersection, which simply coincides with the template for the co-Büchi objective for this example. In general, such a composition leads to a winning strategy templates for the combined objective, if no conflict arises between templates, e.g., the only remaining live action set of a state contains a co-live action. We formalize this notion in \cref{sec:offlineA}-\cref{sec:offlineB}. We empirically demonstrate in \cref{sec:experiments} that such conflicts are rare in practice, and composition is possible in a large number of cases.

\smallskip
\noindent\textbf{Online Adaptability.}
Due to their permissive structure, strategy templates are useful for optimizing the robot's behavior based on additional information available at runtime. To illustrate this, assume that $\robotE$ has a randomized strategy in the game which uses a fixed probability distribution over actions. While this probability distribution is not known to $\robotC$, $\robotC$ can observe the actions chosen by $\robotE$ \emph{at runtime}. $\robotC$ can therefore \enquote{estimate} $\robotE$'s action distribution and adapt its own action distribution \emph{online} s.t.\ it is still following the strategy template but also optimizing its action choices to reach $S_e$ efficiently. This idea was implemented in our prototype tool for the robot example in \cref{fig:grid_map} and the experimental results are depicted in \cref{fig:robot_exp}. We see that $\robotC$ reaches $S_e$ much faster when adapting its strategy.

	\section{Computing Winning Strategy Templates}\label{sec:computation}
	This section contains the main contribution of this paper which is the formalization and construction of \emph{concurrent} safety templates (\cref{sec:safety}), live-group templates (\cref{sec:buechi}) and co-live templates (\cref{sec:cobuechi}). Similar to the construction of these templates for turn-based games, we extract them from the symbolic computation of the winning region for concurrent safety, Büchi and co-Büchi games, respectively.  To formalize this extraction, we present additional preliminaries in \cref{sec:prelimMucalc}.

\subsection{Preliminaries}\label{sec:prelimMucalc}
Observing space constrains, we refer the reader to \cite{kozen1983results} for an introduction to $\mu$-calculus, and to \cite{AnandMNS23} for its particular use in the context of template computations. To ensure a minimum level of self-containedness, this section recalls the set-transformers needed to define the symbolic fixed-point algorithms for  concurrent safety, Büchi and co-Büchi games from \cite{AlfaroHenzinger_Concurrent2000}.

Given a concurrent game $(\gamegraph,\spec)$, a state $v\in V$ and transition distributions $\distri\in \probdistributions(\actions_i(v))$ for players $i\in\{1,2\}$, the one-round probability of reaching a set $X \subseteq \vertex$ from $v$ is defined by 
\begin{equation*}
\orprob{v}{\distro}{\distrt}{X}:=\sum_{a \in \actiono(v)}^{} \sum_{b \in \actiont(v)}^{} \sum_{\{v' \in X | \delta(v,a,b) = v'\}}^{} \distro(a) \distrt(b).
\end{equation*}
Then we define three types of \emph{predecessor operators}:
\begin{compactitem}
    \item $\textsf{pre}_1 : 2^{\vertex} \rightarrow 2^{\vertex}$ s.t.\ $v\in \pre{1}{X}$ if
    \[
    \exists \distro \in \probdistributions(\actiono(v)).\forall \distrt \in \probdistributions(\actiont(v)).\orprob{v}{\distro}{\distrt}{X} = 1.
    \]

    \item $\textsf{Apre}_1 : 2^{\vertex} \times 2^{\vertex} \rightarrow 2^{\vertex}$ s.t.\ $v \in \Apre{1}{Y}{X}$ if
    \[
    \exists \distro \in \probdistributions(\actiono(v)).\forall \distrt \in \probdistributions(\actiont(v)).\left[\orprob{v}{\distro}{\distrt}{Y} = 1 \wedge \orprob{v}{\distro}{\distrt}{X} > 0\right]
    \]
    \item $\textsf{AFpre}_1 : 2^{\vertex} \times 2^{\vertex} \times 2^{\vertex} \rightarrow 2^{\vertex}$ s.t.\ $v \in \AFpre{1}{Z}{Y}{X}$ if for some $\beta > 0$,
    \begin{equation*}
    \exists \distro \in \probdistributions(\actiono(v)).\forall \distrt \in \probdistributions(\actiont(v)).
    \begin{pmatrix*}[l]
     \orprob{v}{\distro}{\distrt}{Z} = 1\\
     \wedge \orprob{v}{\distro}{\distrt}{X} \geq \beta \orprob{v}{\distro}{\distrt}{\neg Y} \geq 0
    \end{pmatrix*}
    \end{equation*}
\end{compactitem}

The predecessor operator $ \pre{1}{X} $ computes the set of states from which \p{1} can ensure reaching $X$ with probability 1. The predecessor operator $\Apre{1}{Y}{X}$ computes the set of states from which \p{1} can ensure to stay in $Y$ almost surely and reach $X$ with positive probability. Finally, the predecessor operator $\AFpre{1}{Z}{Y}{X}$ computes the set of states from which \p{1} can ensure to stay in $Z$ almost surely and if it can leave $Y$ with positive probability, then it should also reach $X$ with some positive probability.

In addition, we utilize two action-set functions from \cite[Sec.6]{AlfaroHenzinger_Concurrent2000}, namely $A^v_Y : 2^{\actiont(v)} \rightarrow 2^{\actiono(v)}$ and  $B^{v}_{X}:2^{\actiono(v)} \rightarrow 2^{\actiont(v)}$ defined s.t.\
\begin{align}
 \safeactions{Y}{v}{\subactiont}&:= \{ \acto \in \actiono(v) \mid \forall \actt \in \actiont(v) ~.~ \transition(v,\acto,\actt) \notin Y \Rightarrow \actt \in \subactiont\},\label{equ:defAset}\\
 \forwardingactions{X}{v}{\subactiono}&:= \{ \actt \in \actiont(v) \mid \exists \acto \in \subactiono \cdot \delta(v,\acto,\actt) \in X \}\label{eq:defBset}
\end{align}
Intuitively, the set $\safeactions{Y}{v}{\subactiont}$ captures the set of actions for $\p{1}$ at state $v$ which are guaranteed to lead to a state in $Y\subseteq V$, under the assumption that $\p{2}$ avoids the actions in $\subactiont$. 
The set $\forwardingactions{X}{v}{\subactiono}$ captures the set of actions for $\p{2}$ at state $v$ for which \p{1} can ensure reaching a state in $X\subseteq V$ with positive probability.
It follows from \cite[Sec.6]{AlfaroHenzinger_Concurrent2000} that \textsf{AFpre} can be computed via the equivalence 
$v \in \AFpre{1}{Z}{Y}{X} \iff \nu \subaction . (\safeactions{Z}{v}{\emptyset} \cap \safeactions{Y}{v}{\forwardingactions{X}{v}{\subaction}}) \neq \emptyset$.

	\subsection{Concurrent Safety Templates}\label{sec:safety}
	This section introduces (maximally permissive) safety templates for concurrent games in direct analogy to safety templates for turn-based games.

\smallskip
\noindent\textbf{Template Formalization.}
As illustrated in \cref{sec:problem} and similar to their turn-based counterparts, concurrent safety templates are defined via a function $\funcSafe : V \rightarrow 2^{\Gamma_1(V)}$ specifying a set of \emph{unsafe actions} for $\p{1}$ at each state and collecting all strategies which avoid actions in $\funcSafe$. I.e., their support\footnote{Concurrent safety games allow deterministic strategies. We still define safety templates over randomized strategies to ease their combination with other templates.} never contains unsafe actions, as formalized next.

\begin{definition}
	Let $\gamegraph$ be a concurrent game graph. Given a function $\funcSafe : V \rightarrow 2^{\Gamma_1(V)}$, the \emph{concurrent safety template} $\templatesafe(\funcSafe)$ is defined by
	\begin{equation}\label{equ:PiSet}
	 \templatesafe(\funcSafe) := \{\strat_1 \mid \forall \play \in \plays(\strat_1)\cdot \forall i\geq 0\cdot \support(\strat_1(\play_{\leq i})) \cap \funcSafe(\play_i) = \emptyset\}.
	\end{equation}
\end{definition}
The remainder of this section shows how $\funcSafe$ can be computed s.t.\  $\templatesafe$ is (almost surely) winning
in a safety game $(\gamegraph,\Box I)$, and therefore sound. %

\smallskip
\noindent\textbf{Template Computation.}
It is shown in \cite{AlfaroHenzinger_Concurrent2000} that the (almost sure) wining region for $\p{1}$ in a concurrent safety game $(G,\Box \targetSet)$ can be computed by a symbolic fixed-point formula %
 which is identical to the fixed-point formula for \emph{turn-based} safety games, i.e.\
 $\wino(\Box \targetSet)= \nu X. \left(\pre{1}{X}\cap \targetSet\right)$.
We use $\SafeAlgo{\gamegraph,\targetSet}$ to denote the algorithm which implements this formula. 

Now recall the action set function $A^v_Y : 2^{\actiont(v)} \rightarrow 2^{\actiono(v)}$ from \cref{equ:defAset}
which captures the set of actions for $\p{1}$ at state $v$ which are guaranteed to lead to a state in $Y\subseteq V$, under the assumption that $\p{2}$ avoids the actions in $\subactiont$.
As $\subactiont$ is the only actions which $\p{2}$ could play to avoid remaining in the \emph{safe} set $Y$, we can set $\subactiont=\emptyset$ and $Y=\wino(\LTLalways \targetSet)$ in \eqref{equ:defAset} to obtain the set of actions for $\p{1}$ which ensure that the play stays within $\wino(\LTLalways \targetSet)$ against all possible actions of $\p{2}$. With this, we obtain the following construction.

\begin{restatable}{theorem}{restatesafety}\label{thm:safety}
	Let $(\gamegraph, \Box \targetSet)$ be a \emph{concurrent safety game} with $\wino=\wino(\Box\targetSet)$ and
	\begin{equation}\label{equ:SafeTempCompute}
	 \funcSafe(v) := 
	 \begin{cases}
	  \actiono(v) \setminus  \safeactions{\wino}{v}{\emptyset} & \text{if}~v\in \wino,\\
	  \emptyset& \text{if}~v\notin \wino\\
	 \end{cases}.
	\end{equation} 
	Then $\templatesafe(\funcSafe)$ is (almost surely) winning in $(\gamegraph, \Box \targetSet)$ and maximally permissive.
\end{restatable}
\noindent We call $\funcSafe$ the \emph{safety function} and refer to its computation via $\safeTemp(\gamegraph, \targetSet)$.

	\subsection{Concurrent Live-Group Templates}\label{sec:buechi}
	This section introduces liveness templates for concurrent games in direct analogy to liveness templates for turn-based games. %
In turn-based games, Anand et al.~\cite{AnandNS_PeSTels2023} formalize a template for liveness via live groups, each containing a set $H$ of \p{1} edges. Whenever \p{1}'s strategy ensures that at least one edge from a group $H$ is taken infinitely often if one source state $\texttt{src}(H)$ is seen infinitely often along a play, progress towards a goal state is guaranteed always again. The computation of live groups for turn-based games follows the steps of a classical fixed-point algorithm for \buchi games. Intuitively, in each iteration $i$ of the fixed-point algorithm, all edges that \p{1} needs to \enquote{actively} choose to make progress to states added in earlier iterations, form one live group $H$. 

While conceptually similar, the data structure needed to realize the same idea in concurrent games is more complicated and its extraction from the corresponding fixed-point algorithm for concurrent Büchi games from \cite{AlfaroHenzinger_Concurrent2000} more intricate. %

\smallskip
\noindent\textbf{Template Formalization.}
There are two main reasons why a different data structure is needed to formalize concurrent liveness templates. First, we observe that due to the action-based semantics of concurrent games, we can not indirectly define the set of source states $\texttt{src}(H)$ via an edge set $H$. We therefore keep a separate list of source states $\livePartitions$ and live actions $\funcLive$. A second, more intricate, difference lies in the definition of $\funcLive$. Due to the synchronous semantics of concurrent games, progress can only be ensured by a randomized strategy which assigns a positive probability to all actions of \p{1} that might lead to progress in case \p{2} chooses the `right' action. We therefore need to remember sets of \p{1} actions per `progress enabling' \p{2} action in every state $v\in\vertex$, leading to a function $\funcLive: \vertex \rightarrow 2^{2^{\actiono(\vertex)}}$. %

Intuitively, a concurrent liveness template must then ensure that for each state $v$ and for every opponent's action $b\in\actiont(v)$ s.t.\ progress might be made, a positive probability is assigned to all actions $\subaction^b \subseteq \actiono(v)$ which would enable this progress. %
Unfortunately, however, it is not sufficient for $\p{1}$ to assign a positive probability to these subsets. It must also be ensured that the probability of taking this action does not vanish over time, to enable infinitely many visits to the target set. %
This is formalized below.
\begin{definition}
	Let $\gamegraph$ be a concurrent game. Given a function  $\funcLive: \vertex \rightarrow 2^{2^{\actiono(\vertex)}}$ and a partition $\livePartitions \subseteq 2^\vertex$, the \emph{concurrent live-group template} $\templatelive ( \funcLive, \livePartitions)$ is defined s.t.\ $\strat_1\in\templatelive ( \funcLive, \livePartitions)$ iff for all 
	$\play \in \plays(\strato)$ and for all $\livePartition \in \livePartitions$ it holds that 
	\begin{equation}\label{equ:PiSetLive}
		\textstyle \left(\livePartition \cap \infPlay{\play} \ne \emptyset\right) \Rightarrow 
		\textstyle \left(\sum_{\{i \mid \play_i\in \livePartition\}} \minsupport{\strato(\play_{\leq i} )}{\funcLive(\play_i)} = \infty\right),
	\end{equation}
	\begin{equation}\label{equ:MinSupp}
     \text{where }\minsupport{p}{X} :=\textstyle \min_{\subaction \in X} p(\subaction).
    \end{equation}
\end{definition}

Intuitively, $\funcLive(\play_i)$ is a set of \p{1} action sets and \eqref{equ:PiSetLive} ensures that a strategy can only follow $\templatelive ( \funcLive, \livePartitions)$ if the minimum probability it assigns to all the live action sets in $\funcLive(\play_i)$ add up to infinity. This ensures that the corresponding actions are indeed taken infinitely often in expectation (if the corresponding live state is seen infinitely often). %
The remainder of this section shows how $\funcLive$ and $\livePartitions$ can be computed s.t.\  $\templatelive ( \funcLive, \livePartitions)$ is (almost surely) winning from all states in the winning region of a Büchi game $(\gamegraph,\Box I)$, and therefore sound.

\smallskip
\noindent\textbf{Template Computation.}
Following \cite{AlfaroHenzinger_Concurrent2000} the (almost sure) winning region of $\p{1}$ in a concurrent \buchi game $(\gamegraph, \Box \Diamond \targetSet)$ can be computed by the symbolic fixed-point formula 
	$\wino(\Box \Diamond \targetSet) = \nu Y. \mu X. ((\neg \targetSet \cap \Apre{1}{Y}{X})) \cup (\targetSet \cap \pre{1}{Y})))$,
and $\BuchiAlgo{\gamegraph, \targetSet}$ denotes the algorithm which implements this formula. %

As commonly done for strategy extraction, we consider the last iteration of $\BuchiAlgo{\gamegraph, \targetSet}$ (which computes $Y=\wino$), and define the sets of states $X_0 \subseteq X_1 \subseteq \ldots X_k = X_{k+1} = \wino$ s.t.\ $X_1=\targetSet \cap \pre{1}{\wino}$ and $X_{i+1} = (\neg \targetSet \cap \Apre{1}{\wino}{X_i}) \cup X_1$ for $i \ge 1$. %
We further define $\targetSet_i := X_i \setminus X_{i-1}$ for $i \ge 1$.
Now, we intuitively define $\funcLive$ and $\livePartitions$ s.t.\ whenever some $v\in \targetSet_i$ has been seen infinitely often in the \emph{play} $\play$, then $\p{1}$ should assign probabilities in a way, to be able to reach $\targetSet_{i-1}$ infinitely many times. 
Formally, given the above definitions, we define $\livePartitions:=\{\targetSet_1,\targetSet_2, \ldots \targetSet_k\}$, and for $v\in I_i$ we have
 \begin{align}\label{equ:LiveFunc}
 \funcLive(v):=\{\funcLive^{\actt}(v)\}_{\actt \in \actiont(v)}
&&\text{s.t.}~\funcLive^\actt(v):=\{\acto \in \actiono(v) \setminus \funcSafe(v) \mid \transition(v,\acto,\actt) \in X_{i-1}\},\notag
 \end{align}
where $\funcSafe(v)$ is defined as in \eqref{equ:PiSet}, but w.r.t.\ $\wino=\wino(\LTLalways \LTLeventually \targetSet)$. We call $(\funcLive,\livePartitions)$ the \emph{live-group function} and refer to their outlined computation via $\liveTemp(\gamegraph,\targetSet)$.

\smallskip
\noindent\textbf{Soundness.}
Following \cite{AnandNS_PeSTels2023}, we show that concurrent liveness templates computed for concurrent Büchi games are indeed sound. Intuitively, every winning strategy in such a game must (i) ensure that the play always stays in the winning region and (ii) always makes progress towards the Büchi states. While (i) can be ensured by concurrent safety templates, (ii) is ensured by concurrent liveness templates. This is formalized by the following theorem (see \cref{app:algorithms} for the full algorithm).

\begin{restatable}{theorem}{restatebuchi}\label{thm:buchi}
	Let $(\gamegraph, \LTLalways\LTLeventually \targetSet)$ be a \emph{concurrent \buchi game} with winning region $\wino$, safety function $\funcSafe=\safeTemp(\gamegraph,\wino)$ and live-group functions
 $(\funcLive, \livePartitions)=\liveTemp(\gamegraph, \targetSet)$, then $\template = \templatesafe(\funcSafe) \cap \templatelive(\funcLive, \livePartitions)$ 
is a \emph{winning strategy template} for $(\gamegraph, \LTLalways\LTLeventually \targetSet)$.
\end{restatable}

	We will show in \cref{cor:completeness} that the template $\template = \templatesafe(\funcSafe) \cap \templatelive(\funcLive, \livePartitions)$ is complete, i.e., if there exists a winning strategy for $\p{1}$ in the game $\game=(\gamegraph,\LTLalways \LTLeventually \targetSet)$, then there exists a strategy that follows the template $\template$.
	However, the template is not maximally permissive, i.e., there may exist winning strategies for $\p{1}$ that do not follow the template (see \cref{app:non_maximally} for an example).

	\subsection{Concurrent Co-Live Templates}\label{sec:cobuechi}
	This section introduces co-live templates for concurrent games in direct analogy to co-live templates for turn-based games. While the intuition behind their definition directly carries over as expected, the extraction of \emph{sound} co-live templates from the fixed-point computation of the winning region of a concurrent co-Büchi game is unfortunately not as straightforward. %
Before we discuss this challenge in more detail, we formalize co-live templates.

\smallskip
\noindent\textbf{Template Formalization.}
As illustrated in \cref{sec:problem} and similar to their turn-based counterparts, concurrent co-live templates are defined via a co-live function  $\funcCoLive: \vertex \rightarrow 2^{\actiono(\vertex)}$ specifying the set of actions that \p{1} is only allowed to take finitely often in $v$. To ensure that all strategies $\strato$ which follow a co-live template take such actions only a \emph{finite} number of times, the sum of probabilities associated to these actions via $\strato$ over time must be bounded. %
\begin{definition}
	Let $\gamegraph$ be a concurrent game graph. Given a co-live function $\funcCoLive: \vertex \rightarrow 2^{\actiono(\vertex)}$, the \emph{concurrent co-live template} $\templatecolive (\funcCoLive)$ is defined s.t.\ $\strat_1\in \templatecolive (\funcCoLive)$ iff for all $\play \in \plays(\strato)$ holds that 
	\begin{equation}\label{equ:PiSetCoLive2}
		v \in \infPlay{\play} \Rightarrow 
		\textstyle \sum_{\{i \mid \play_i=v\}} \strato(\play_{\leq i})(\funcCoLive(\play_i)) \ne \infty.
	\end{equation}
\end{definition}

\smallskip
\noindent\textbf{Template Computation.}
Following \cite{AlfaroHenzinger_Concurrent2000}, the (almost sure) winning region of $\p{1}$ in a concurrent \cobuchi game $(\gamegraph, \Diamond \Box \targetSet)$ can be computed by the symbolic fixed-point formula
\begin{equation}\label{eq:cobuchi}
	\wino(\Diamond \Box \targetSet) = \nu Z. \mu X. \nu Y.\left(( \targetSet \cap \AFpre{1}{Z}{Y}{X}) \cup ( \neg \targetSet \cap \Apre{1}{Z}{X})\right),
\end{equation}
where $\CobuchiAlgo{\gamegraph, \targetSet}$ denotes the algorithm implementing this function. It is interesting to note that \eqref{eq:cobuchi} is a three-nested fixed-point formula, while the computation of $\wino(\Diamond \Box \targetSet)$ in a turn-based game only requires two nestings. 

\begin{algorithm}[t]
\caption{$\cobuchiTemp(\gamegraph, \targetSet)$}\label{alg:cobuechi_temp}
 \footnotesize
\begin{algorithmic}[1]
	\Require A game graph $\gamegraph$, and a subset of states $ \targetSet$
	\Ensure A safety function $ \funcSafe $ and a liveness tuple $ (\funcLive, \livePartitions) $ and a co-live tuple $ (\funcCoLive) $
	\State $Z \gets \CobuchiAlgo{\gamegraph,\targetSet}$; $ \funcSafe\gets \safeTemp(\gamegraph,Z) $;\label{line:alg:cobuechi_temp:safecolive}
	\State $X \gets \SafeAlgo{\gamegraph, \targetSet}$; $\funcCoLive \gets \safeTemp(\gamegraph,X)$;\label{line:alg:cobuechi_temp:funcCoLive}
	\For{$ v \in X \cup (\vertex\setminus Z)$} $\funcLive(v) \gets \{\actiono(v)\setminus \funcSafe(v)\}$; \EndFor\label{line:alg:cobuechi_temp:liveoutside}
	\While{$True$}
	\State $Y' \gets \vertex$;
	\Repeat
		\State $Y \gets Y'$; $Y' \gets (I \cap \AFpre{1}{Z}{Y}{X}) \cup (\neg I \cap \Apre{1}{Z}{X})$;
	\Until{$Y' == Y$}
	\If{$X == Y$} \textbf{break} \EndIf
	\For{$v \in Y \setminus X$}
	\State $\funcLive(v) \gets \emptyset$;
	\For{$ \actt \in \actiont(v)$}
	$\funcLive(v) \gets \funcLive(v) \cup\{\{ \acto \in \actiono(v) \setminus \funcSafe(v) \mid \delta(v,\acto,\actt) \in X\} \}$;\label{line:alg:cobuechi_temp:removeunsafe}
	\EndFor
	\EndFor
	\If{$Y \setminus X \ne \emptyset$}
	$\livePartitions \gets \livePartitions \cup \{Y \setminus X\}$;
	\EndIf
	\State $X \gets Y$;
	\EndWhile
	\State \Return $ (\funcSafe, (\funcLive, \livePartitions), \funcCoLive) $
\end{algorithmic}
\end{algorithm}

As in the Büchi case, we can extract a winning strategy template from the last iteration of $\CobuchiAlgo{\gamegraph, \targetSet}$ (computing $Z^*=\wino$) as detailed in \cref{alg:cobuechi_temp}. As before, we first compute the winning region $Z^*=\wino$ and ensure that the play stays in this set via a safety function $\funcSafe$ (line 1).

In order to understand line 2 of \cref{alg:cobuechi_temp}, recall that $Z^*$ can be written as an increasing sequence of sets $X_1 \subset X_2 \subset X_3 \subset ... \subset X_k = X_{k+1} = Z^*=\wino$. It can now be shown that it follows from the definition of $\textsf{AFpre}$ in \cref{sec:prelimMucalc} that $\AFpre{1}{Z^*}{Y}{\emptyset}$ reduces to $\pre{1}{Y}$ while it directly follows that $\Apre{1}{Z^*}{\emptyset} =\emptyset$ for all $Y \subseteq V$. We can therefore conclude that $X_1=Y^*=\nu Y.\left[\targetSet \cap \pre{1}{Y}\right] = \wino(\Box \targetSet)=\SafeAlgo{\gamegraph,\targetSet}$. I.e., $X_1$ are the states from which \p{1} has a strategy to keep the game in $I$ indefinitely. 
However, due to the definition of $Z^*$, \p{1} is allowed to leave $X_1$ a finite number of times into $Z^*$, as $\textsf{AFpre}$ ensures that the play can be forced to again make progress towards $X_1$ almost surely. We therefore collect all actions leaving $X_1$ into $Z^*$ into the co-live function $\funcCoLive$ (line 2 in \cref{alg:cobuechi_temp}).

Finally, to ensure the discussed progress towards $X_1$, we additionally need a live-group template $(\funcLive,\livePartitions)$. This construction is very similar to the Büchi case, but using the particular fixed point in \eqref{eq:cobuchi} (line 4-19 in \cref{alg:cobuechi_temp}).

	In conclusion, a winning strategy template for concurrent \cobuchi games consists of a safety template (to avoid leaving winning region), a co-live template (to avoid leaving $X_1$ infinitely many times), and a live-group template (to ensure reaching $X_1$ eventually). This is formalized in the following theorem.
	
	\begin{restatable}{theorem}{restatecobuchi}\label{thm:cobuchi}
	Let $(\gamegraph, \LTLeventually \LTLalways\targetSet)$ be a \emph{concurrent \cobuchi game} with winning region $\wino$, if $(\funcSafe, (\funcLive, \livePartitions), \funcCoLive) = \cobuchiTemp(\gamegraph, \targetSet)$, then $\template = \templatesafe(\funcSafe) \cap \templatelive(\funcLive, \livePartitions) \cap \templatecolive(\funcCoLive)$ is a \emph{winning strategy template} for  $(\gamegraph, \LTLeventually \LTLalways\targetSet)$.
	\end{restatable}

	While we will show in \cref{cor:completeness} that the template $\template = \templatesafe(\funcSafe) \cap \templatelive(\funcLive, \livePartitions) \cap \templatecolive(\funcCoLive)$ is complete for the game $\game=(\gamegraph,\LTLeventually\LTLalways \targetSet)$, the template is not maximally permissive (see \cref{app:non_maximally} for an example).

	\begin{remark}
	 We remark that winning strategy templates in turn-based co-Büchi games can be expressed solely by safety and co-liveness templates. This is due to the fact that live-group templates which are needed to ensure progress towards the safe invariant set can be translated into co-live templates for turn-based co-Büchi games, as progress needs to be made only over a \emph{finite} time horizon. For concurrent games, this is unfortunately not possible, as the concurrent versions of live-group and co-live templates are not dual in this case.
	\end{remark}

	\section{Conflict-freeness and Strategy Extraction}\label{sec:offlineA}
	This section discusses how to efficiently extract a strategy that follows a given strategy template.
Although in many cases, as discussed in \cref{sec:problem}, one can easily construct a strategy that follows a given strategy template.
However, it is not always the case that a strategy following a given strategy template even exists.
For example, consider a strategy template $\templatesafe(\funcSafe)$ that disallow all the available actions from a state $v$, i.e., $\funcSafe(v) = \actiono(v)$.
In this case, there is no strategy that follows the template.
We call situations like this \emph{conflicts} in the strategy template.
To formalize this notion, we adapt the concept of \emph{conflict-freeness} of strategy templates from \cite{AnandNS_PeSTels2023} to concurrent games.

\begin{definition}\label{def:conflict-free}
	Given a game graph $\gamegraph=(V,\actions_1,\actions_2,\transition)$, a strategy template $\template = \templatesafe(\funcSafe) \cap \templatelive(\funcLive, \livePartitions) \cap \templatecolive(\funcCoLive)$ is \emph{conflict-free} if the following conditions hold for every state $v \in V$:
	(i) $\actiono(v) \not\subseteq \funcSafe(v)\cup\funcCoLive(v)$, and 
	(ii) for every $\subaction \in \funcLive(v)$, $\subaction \not\subseteq \funcSafe(v)\cup\funcCoLive(v)$.
\end{definition}
Intuitively, condition (i) ensures that there is at least one action available at every state $v$ that is not disallowed by the safety or co-liveness templates.
Condition (ii) ensures that even without the unsafe and co-live actions, liveness template can still be satisfied, i.e., for every opponent's action, there is at least one live action available at $v$ that is not disallowed by the safety or co-liveness templates.

With the above intuition, one easy way to extract a memoryless strategy that follows a given conflict-free strategy template is the following: for every state, remove all unsafe and co-live actions, and then, assign positive probabilities to the remaining actions.
This strategy is well-defined as condition~(i) of \cref{def:conflict-free} ensures that there is at least one action available at every state after removing unsafe and co-live actions.
Trivially, this strategy follows the safety and co-liveness templates.
Moreover, this strategy also follows the liveness template as condition~(ii) of \cref{def:conflict-free} ensures that for every opponent's action, at least one live action has been assigned a positive probability.
Hence, the strategy indeed follows the strategy template, giving us the following result.

\begin{restatable}{theorem}{restatestrategyextraction}\label{thm:strategy-extraction}
	Given a game graph $\gamegraph=(V,\actions_1,\actions_2,\transition)$ and a conflict-free strategy template $\template$, one can extract a memoryless strategy that follows the template in time $\bigO(\abs{V}\cdot \abs{\actions_1}\cdot \abs{\actions_2})$.
\end{restatable}

Similar to \cite{AnandNS_PeSTels2023}, we can also show that all the algorithms we presented in \cref{sec:computation} produce conflict-free strategy templates. 

\begin{restatable}{theorem}{restateallconflictfree}\label{thm:all-conflict-free}
	The procedures $\safeTemp$, $\buchiTemp$, and $\cobuchiTemp$ always return conflict-free strategy templates.
\end{restatable}

As a consequence of \cref{thm:all-conflict-free}, we also get the completeness of our algorithms in \cref{sec:computation}, i.e., for every state in the winning region, we can extract a winning strategy that follows the computed strategy template.
\begin{corollary}\label{cor:completeness}
	Given the premises of \cref{thm:buchi} (resp. \cref{thm:cobuchi}), for every state $v\in\wino$, there exists a winning strategy $\strato$ for $\p{1}$ from $v$ that follows the strategy template $\template$ computed by \cref{alg:buchiTemp} (resp. \cref{alg:cobuechi_temp}).
\end{corollary}

	\section{Composition of Strategy Templates}\label{sec:offlineB}
	With the conflict-freeness property, we can also now combine different strategy templates by intersecting them, and still be assured that there exists a strategy that follows the combined template.
This is particularly useful when we want to combine multiple objectives, as we can compute the strategy template for each objective separately, and then intersect them to get a combined strategy template.
Let us start by defining the combination of two strategy templates.

\begin{definition}\label{def:composition}
    For a game graph $\gamegraph = (V, \actiono, \actiont, \transition)$,
    the combination of strategy templates $\template_i = \templatesafe(\funcSafe_i) \cap \templatelive(\funcLive_i, \livePartitions_i) \cap \templatecolive(\funcCoLive_i)$ for $i \in \{1,2\}$, is defined as the strategy template $\template_1 \cap \template_2 = \templatesafe(\funcSafe) \cap \templatelive(\funcLive, \livePartitions_1 \cup \livePartitions_2) \cap \templatecolive(\funcCoLive)$ with $\funcSafe(v) = \funcSafe_1(v) \cup \funcSafe_2(v)$, $\funcLive(v) = \funcLive_1(v) \cup \funcLive_2(v)$, and $\funcCoLive(v) = \funcCoLive_1(v) \cup \funcCoLive_2(v)$.
\end{definition}

Intuitively, the combined strategy template marks an action as unsafe if it is unsafe in either of the individual templates, and similarly for co-live actions.
For live actions, the combined template adds all the groups of live actions from both templates and also considers both partitions of the individual templates. This ensures that a strategy follows the combined template if and only if it follows both individual templates.
With this definition, whenever the combined strategy template is conflict-free, we can be assured that there exists a strategy that follows both individual templates,
leading to the following result.
\begin{restatable}{theorem}{restatecomposition}\label{thm:composition}
    Let $(\gamegraph, \spec_1)$ and $(\gamegraph, \spec_2)$ be concurrent games with corresponding winning regions $\win_1$ and $\win_2$, and winning strategy templates $\template_1$ and $\template_2$. If the combined strategy template $\template_1 \cap \template_2$ is conflict-free, then it is a winning strategy template for the game $(\gamegraph, \spec_1 \land \spec_2)$ with winning region $\win_1 \cap \win_2$. 
\end{restatable}

Note that, in general, the winning region of the game $(\gamegraph, \spec_1 \land \spec_2)$ is not necessarily equal to $\win_1 \cap \win_2$. However, if the combined strategy template is conflict-free, then the winning region of the conjoined game is indeed $\win_1 \cap \win_2$.

	\section{Experiments}\label{sec:experiments}
	\noindent\textbf{Implementation and Setup.} Our algorithms are implemented in a Python tool, called \toolname, capable of handling Büchi and co-Büchi games, hence Parity games with two colors.
We are unaware of large-scale benchmarks for \emph{concurrent} games. Thus, we converted \emph{turn-based} Parity games from the \textsc{SYNTCOMP} benchmarks \cite{JacobsPABCCDDDFFKKLMMPR24} to concurrent games for evaluation purposes. Specifically, we considered smaller models which are alternating (i.e. \p{1} states are always followed by \p{2} states and vice versa) and can be converted to Büchi (resp. co-Büchi) games. Our conversion first merges each \p{1} transition with the transitions of the respective successor \p{2} state. Afterwards, it removes all \p{2} states and turns the transition-based into state-based winning conditions\footnote{We note that the resulting concurrent game does not capture the same interaction dynamics as the original turn-based game. The conversion procedure and considered models are described in more detail in Appendix~\cref{app:experiments}}. In total we converted 171 games. All experiments were run on a machine with Ubuntu 22.04, an Intel i7-1165G7 CPU and 32GB RAM. Our implementation and experimental data are available in \cite{anonymous_2025_17357029}.

We compared our prototypical implementation with \texttt{PeSTel} \cite{AnandNS_PeSTels2023}, a tool for computing strategy templates in \emph{turn-based} games. The results are given in Appendix~\cref{app:experiments} and indicate that \toolname is noticeably slower. We attribute the slower performance to programming language differences and the inherent difficulty of solving concurrent games, and leave engineering improvements open for future work. For the remainder, we investigate the efficacy of our templates in incremental synthesis and optimizing strategies at runtime.

\begin{figure}[!t]
  \centering
  \hfill
  \begin{minipage}[b]{0.49\textwidth}
    \includegraphics[width=0.99\textwidth]{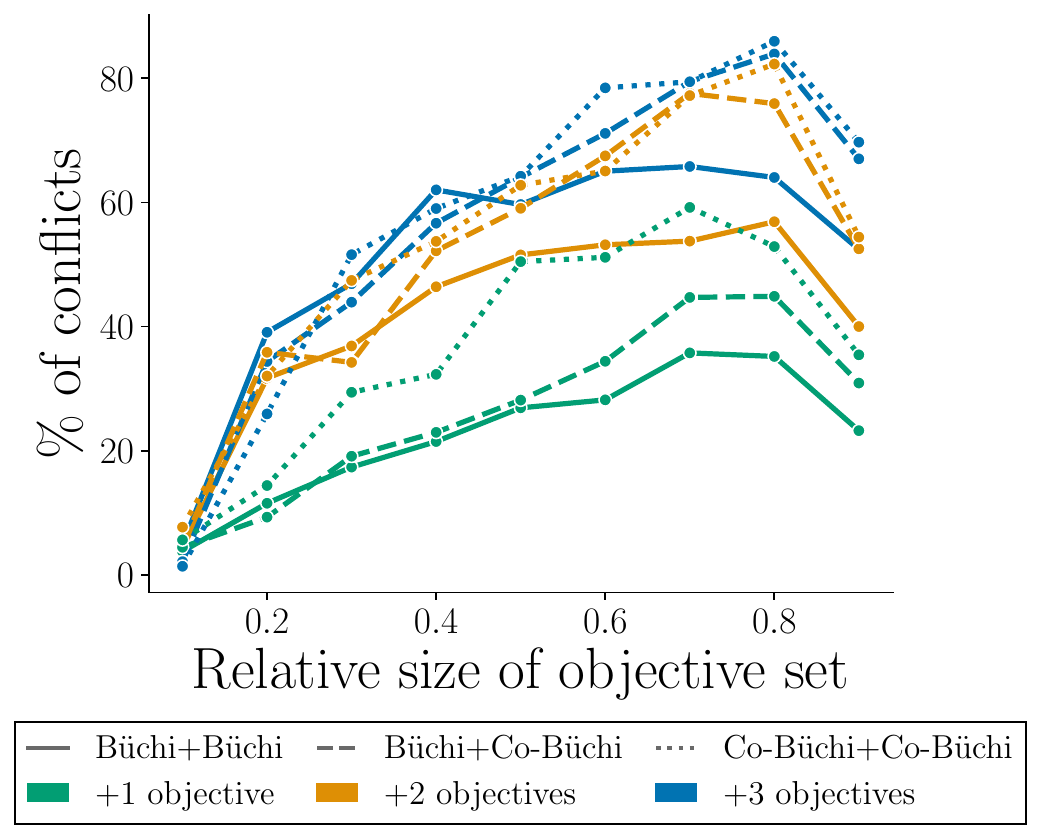}
    \caption{Conflict analysis.}
    \label{figure:lineplot-conflict}
  \end{minipage}
  \hfill
      \begin{minipage}[b]{0.49\textwidth}
    	\centering
	\includegraphics[width=0.99\textwidth]{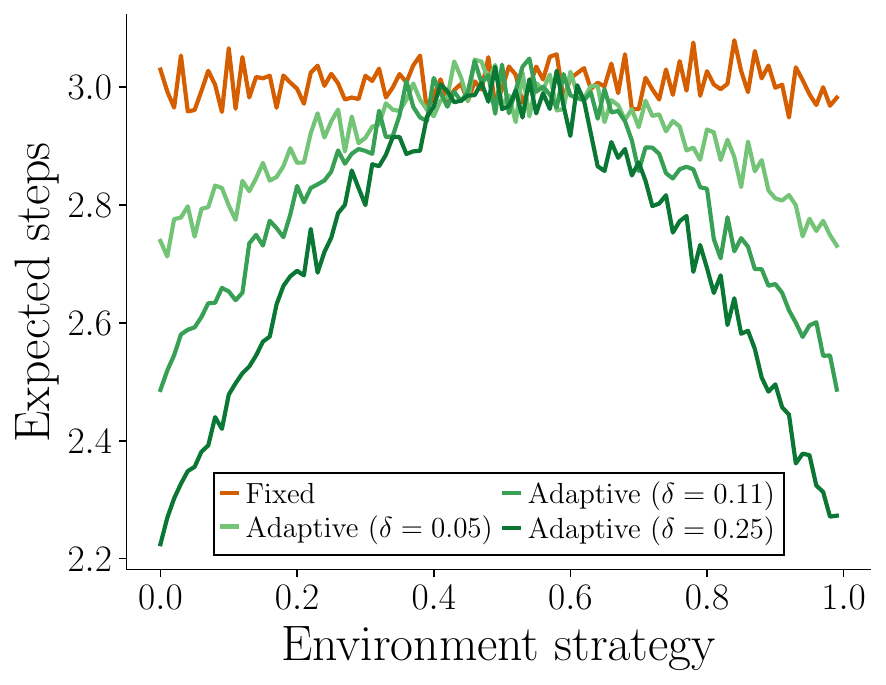}
	\caption{Expected steps to reach $S_e$.}
	\label{fig:robot_exp}
  \end{minipage}
  \vspace{-0.5cm}
\end{figure}

\smallskip

\noindent\textbf{Incremental Synthesis.} We study the ability of our approach to handle newly arriving objectives. Starting with a template for an initial objective, we are interested in the amount of conflicts that occur when combining the existing template with templates of new objectives. We selected four converted games with at least $20$ states and for which the winning regions of the initial objectives are non-empty. For each game, we then introduced $k \in \{1,2,3\}$ different additional objectives of the same type (Büchi or co-Büchi), each of size $s \in \{10\%, 20\%, \ldots, 90\%\}$ (relative to the number of states). For every configuration $(s, k)$ we considered $100$ random samples and reported the relative number of conflicts when combining the existing template with the templates of the new objectives. The results are shown in \cref{figure:lineplot-conflict}, where the number of additional objectives $k$ is indicated by colors and the objective types by line styles.
As expected, more conflicts occur as the number of additional objectives increases. Most conflicts occur for larger objective sets ($60\%$ to $80\%$). Intuitively, small sets result in smaller winning regions and thus fewer states where conflicts can potentially occur. Less conflicts arise for the largest set, since the objectives become easier to satisfy. Remarkably, regardless of the objective size, we can \emph{often add additional Büchi templates} to an existing Büchi template without conflicts ($\leq 36\%$ conflicts for $+1$ objective). Conflicts are more prevalent when co-Büchi objectives are involved, but remain reasonable for smaller objectives. We provide more details and plots in Appendix~\cref{app:experiments}.

\smallskip

\noindent\textbf{Runtime Optimization.} Finally, we showcase the application of our templates in runtime strategy optimization on the robot example (\cref{fig:robot_all}). The robot $R_C$ reaches a green cell alone almost surely when assigning positive probabilities to actions that lead to green cells, as prescribed by our strategy template. Observe that strategies can differ in the \emph{expected time} it takes for the robot to reach $S_e$. %
Conceptually, the template enables the robot to \emph{change its strategy at runtime}, while ensuring it reaches $S_e$ almost surely. Using this flexibility, the robot can adapt its strategy, among the ones prescribed by the template, to minimize the expected time. Assuming a fixed environment strategy, we consider a simple adaptive approach that updates the robot's strategy based on the observed environment actions. Initially, the robot assigns equal probability to both actions leading to the green cells. As the game progresses, the robot observes the environment actions and adjusts the probabilities by $\delta$ to favour actions that are more likely to succeed, while ensuring that the template is still followed.

We simulated the approach $10000$ times for different environment strategies and values for $\delta$. The results are shown in \cref{fig:robot_exp}.
The x-axis corresponds to environment strategies and shows the probability of the environment choosing the \emph{clockwise} move. The y-axis depicts the expected time for the robot to reach $S_e$. The fixed strategy assigns equal probability to both actions and generally requires three steps to reach $S_e$. Depending on $\delta$ and the environment strategy, the robot is able to noticeably reduce the number of steps when following the adaptive approach. This shows the potential for our template-based approach in scenarios where dynamic adaptation offers an advantage.

\medskip

\noindent\textit{Data Availability Statement.} The artefact containing the models, experimental data, implementation, and scripts for reproducing the experiments can be found at \url{https://doi.org/10.5281/zenodo.18187322}. The latest version of the artefact can be found at \url{https://doi.org/10.5281/zenodo.17357028} \cite{anonymous_2025_17357029}.

\bibliographystyle{splncs04}
\bibliography{bib}

@INPROCEEDINGS{AlfaroHenzinger_Concurrent2000,

  author={de Alfaro, L. and Henzinger, T.A.},

  booktitle={Proceedings Fifteenth Annual IEEE Symposium on Logic in Computer Science (Cat. No.99CB36332)}, 

  title={Concurrent omega-regular games}, 

  year={2000},

  volume={},

  number={},

  pages={141-154},

  doi={10.1109/LICS.2000.855763}}

@InProceedings{PhalakarnPH_ICTAC2024,
author="Phalakarn, Kittiphon
and Pruekprasert, Sasinee
and Hasuo, Ichiro",
editor="Anutariya, Chutiporn
and Bonsangue, Marcello M.",
title="Winning Strategy Templates for Stochastic Parity Games Towards Permissive and Resilient Control",
booktitle="Theoretical Aspects of Computing -- ICTAC 2024",
year="2025",
publisher="Springer Nature Switzerland",
address="Cham",
pages="197--214",
isbn="978-3-031-77019-7"
}

@InProceedings{AnandNS_PeSTels2023,
author="Anand, Ashwani
and Nayak, Satya Prakash
and Schmuck, Anne-Kathrin",
editor="Enea, Constantin
and Lal, Akash",
title="Synthesizing Permissive Winning Strategy Templates for Parity Games",
booktitle="Computer Aided Verification",
year="2023",
publisher="Springer Nature Switzerland",
address="Cham",
pages="436--458",
isbn="978-3-031-37706-8"
}

@inproceedings{AnandMNS23,
  author       = {Ashwani Anand and
                  Kaushik Mallik and
                  Satya Prakash Nayak and
                  Anne{-}Kathrin Schmuck},
  editor       = {Sriram Sankaranarayanan and
                  Natasha Sharygina},
  title        = {Computing Adequately Permissive Assumptions for Synthesis},
  booktitle    = {Tools and Algorithms for the Construction and Analysis of Systems
                  - 29th International Conference, {TACAS} 2023, Held as Part of the
                  European Joint Conferences on Theory and Practice of Software, {ETAPS}
                  2022, Paris, France, April 22-27, 2023, Proceedings, Part {II}},
  series       = {Lecture Notes in Computer Science},
  volume       = {13994},
  pages        = {211--228},
  publisher    = {Springer},
  year         = {2023},
  url          = {https://doi.org/10.1007/978-3-031-30820-8\_15},
  doi          = {10.1007/978-3-031-30820-8\_15},
  bibsource    = {dblp computer science bibliography, https://dblp.org}
}

@inproceedings{AnandNS24,
  author       = {Ashwani Anand and
                  Satya Prakash Nayak and
                  Anne{-}Kathrin Schmuck},
  editor       = {S. Akshay and
                  Aina Niemetz and
                  Sriram Sankaranarayanan},
  title        = {Strategy Templates - Robust Certified Interfaces for Interacting Systems},
  booktitle    = {Automated Technology for Verification and Analysis - 22nd International
                  Symposium, {ATVA} 2024, Kyoto, Japan, October 21-25, 2024, Proceedings,
                  Part {I}},
  series       = {Lecture Notes in Computer Science},
  volume       = {15054},
  pages        = {22--41},
  publisher    = {Springer},
  year         = {2024},
  url          = {https://doi.org/10.1007/978-3-031-78709-6\_2},
  doi          = {10.1007/978-3-031-78709-6\_2},
  bibsource    = {dblp computer science bibliography, https://dblp.org}
}

@inproceedings{AlurHMQRT98,
  author       = {Rajeev Alur and
                  Thomas A. Henzinger and
                  Freddy Y. C. Mang and
                  Shaz Qadeer and
                  Sriram K. Rajamani and
                  Serdar Tasiran},
  editor       = {Alan J. Hu and
                  Moshe Y. Vardi},
  title        = {{MOCHA:} Modularity in Model Checking},
  booktitle    = {Computer Aided Verification, 10th International Conference, {CAV}
                  '98, Vancouver, BC, Canada, June 28 - July 2, 1998, Proceedings},
  series       = {Lecture Notes in Computer Science},
  volume       = {1427},
  pages        = {521--525},
  publisher    = {Springer},
  year         = {1998},
  url          = {https://doi.org/10.1007/BFb0028774},
  doi          = {10.1007/BFB0028774},
  bibsource    = {dblp computer science bibliography, https://dblp.org}
}

@inproceedings{AlurH96,
  author       = {Rajeev Alur and
                  Thomas A. Henzinger},
  title        = {Reactive Modules},
  booktitle    = {Proceedings, 11th Annual {IEEE} Symposium on Logic in Computer Science,
                  New Brunswick, New Jersey, USA, July 27-30, 1996},
  pages        = {207--218},
  publisher    = {{IEEE} Computer Society},
  year         = {1996},
  url          = {https://doi.org/10.1109/LICS.1996.561320},
  doi          = {10.1109/LICS.1996.561320},
  bibsource    = {dblp computer science bibliography, https://dblp.org}
}

@InProceedings{Kwiatkowska2020,
author="Kwiatkowska, Marta
and Norman, Gethin
and Parker, David
and Santos, Gabriel",
editor="Lahiri, Shuvendu K.
and Wang, Chao",
title="PRISM-games 3.0: Stochastic Game Verification with Concurrency, Equilibria and Time",
booktitle="Computer Aided Verification",
year="2020",
publisher="Springer International Publishing",
address="Cham",
pages="475--487",
isbn="978-3-030-53291-8"
}

@article{AlurHK02,
  author       = {Rajeev Alur and
                  Thomas A. Henzinger and
                  Orna Kupferman},
  title        = {Alternating-time temporal logic},
  journal      = {J. {ACM}},
  volume       = {49},
  number       = {5},
  pages        = {672--713},
  year         = {2002},
  url          = {https://doi.org/10.1145/585265.585270},
  doi          = {10.1145/585265.585270},
  bibsource    = {dblp computer science bibliography, https://dblp.org}
}

@inproceedings{GutierrezNPW18,
  author       = {Julian Gutierrez and
                  Muhammad Najib and
                  Giuseppe Perelli and
                  Michael J. Wooldridge},
  editor       = {Shuvendu K. Lahiri and
                  Chao Wang},
  title        = {{EVE:} {A} Tool for Temporal Equilibrium Analysis},
  booktitle    = {Automated Technology for Verification and Analysis - 16th International
                  Symposium, {ATVA} 2018, Los Angeles, CA, USA, October 7-10, 2018,
                  Proceedings},
  series       = {Lecture Notes in Computer Science},
  volume       = {11138},
  pages        = {551--557},
  publisher    = {Springer},
  year         = {2018},
  url          = {https://doi.org/10.1007/978-3-030-01090-4\_35},
  doi          = {10.1007/978-3-030-01090-4\_35},
  bibsource    = {dblp computer science bibliography, https://dblp.org}
}

@inproceedings{Brenguier13,
  author       = {Romain Brenguier},
  editor       = {Natasha Sharygina and
                  Helmut Veith},
  title        = {{PRALINE:} {A} Tool for Computing Nash Equilibria in Concurrent Games},
  booktitle    = {Computer Aided Verification - 25th International Conference, {CAV}
                  2013, Saint Petersburg, Russia, July 13-19, 2013. Proceedings},
  series       = {Lecture Notes in Computer Science},
  volume       = {8044},
  pages        = {890--895},
  publisher    = {Springer},
  year         = {2013},
  url          = {https://doi.org/10.1007/978-3-642-39799-8\_63},
  doi          = {10.1007/978-3-642-39799-8\_63},
  bibsource    = {dblp computer science bibliography, https://dblp.org}
}

@article{LomuscioQR17,
  author       = {Alessio Lomuscio and
                  Hongyang Qu and
                  Franco Raimondi},
  title        = {{MCMAS:} an open-source model checker for the verification of multi-agent
                  systems},
  journal      = {Int. J. Softw. Tools Technol. Transf.},
  volume       = {19},
  number       = {1},
  pages        = {9--30},
  year         = {2017},
  url          = {https://doi.org/10.1007/s10009-015-0378-x},
  doi          = {10.1007/S10009-015-0378-X},
  bibsource    = {dblp computer science bibliography, https://dblp.org}
}

@inproceedings{ChengKLB11,
  author       = {Chih{-}Hong Cheng and
                  Alois C. Knoll and
                  Michael Luttenberger and
                  Christian Buckl},
  editor       = {Parosh Aziz Abdulla and
                  K. Rustan M. Leino},
  title        = {{GAVS+:} An Open Platform for the Research of Algorithmic Game Solving},
  booktitle    = {Tools and Algorithms for the Construction and Analysis of Systems
                  - 17th International Conference, {TACAS} 2011, Held as Part of the
                  Joint European Conferences on Theory and Practice of Software, {ETAPS}
                  2011, Saarbr{\"{u}}cken, Germany, March 26-April 3, 2011. Proceedings},
  series       = {Lecture Notes in Computer Science},
  volume       = {6605},
  pages        = {258--261},
  publisher    = {Springer},
  year         = {2011},
  url          = {https://doi.org/10.1007/978-3-642-19835-9\_22},
  doi          = {10.1007/978-3-642-19835-9\_22},
  bibsource    = {dblp computer science bibliography, https://dblp.org}
}

@article{shapley1953stochastic,
  title={Stochastic games},
  author={Shapley, Lloyd S},
  journal={Proceedings of the national academy of sciences},
  volume={39},
  number={10},
  pages={1095--1100},
  year={1953},
  publisher={National Academy of Sciences}
}

@article{EtessamiY10,
  author       = {Kousha Etessami and
                  Mihalis Yannakakis},
  title        = {On the Complexity of Nash Equilibria and Other Fixed Points},
  journal      = {{SIAM} J. Comput.},
  volume       = {39},
  number       = {6},
  pages        = {2531--2597},
  year         = {2010},
  url          = {https://doi.org/10.1137/080720826},
  doi          = {10.1137/080720826},
  bibsource    = {dblp computer science bibliography, https://dblp.org}
}

@inproceedings{AlfaroHK98,
  author       = {Luca de Alfaro and
                  Thomas A. Henzinger and
                  Orna Kupferman},
  title        = {Concurrent Reachability Games},
  booktitle    = {39th Annual Symposium on Foundations of Computer Science, {FOCS} 1998,
                  Palo Alto, California, USA, November 8-11, 1998},
  pages        = {564--575},
  publisher    = {{IEEE} Computer Society},
  year         = {1998},
  url          = {https://doi.org/10.1109/SFCS.1998.743507},
  doi          = {10.1109/SFCS.1998.743507},
  bibsource    = {dblp computer science bibliography, https://dblp.org}
}

@book{gamesbook,
 title     = {Games on Graphs:
    From Logic and Automata to Algorithms},
 author    = {Nathana{\"e}l Fijalkow and
    C. Aiswarya and
    Guy Avni and
    Nathalie Bertrand and
    Patricia Bouyer and
    Romain Brenguier and
    Arnaud Carayol and
    Antonio Casares and
John Fearnley and
    Paul Gastin and
    Hugo Gimbert and
    Thomas A. Henzinger and
    Florian Horn and
    Rasmus Ibsen-Jensen and
    Nicolas Markey and
    Benjamin Monmege and
    Petr Novotn{\"y} and
    Pierre Ohlmann and
    Mickael Randour and
    Ocan Sankur and
    Sylvain Schmitz and
    Olivier Serre and
    Mateusz Skomra and
    Nathalie Sznajder and
    Pierre Vandenhove},
 editor    = {Nathana{\"e}l Fijalkow},
 publisher = {Online},
 year      = {2025},
}

@article{Chatterjee07,
  author       = {Krishnendu Chatterjee},
  title        = {Concurrent games with tail objectives},
  journal      = {Theor. Comput. Sci.},
  volume       = {388},
  number       = {1-3},
  pages        = {181--198},
  year         = {2007},
  url          = {https://doi.org/10.1016/j.tcs.2007.07.047},
  doi          = {10.1016/J.TCS.2007.07.047},
  bibsource    = {dblp computer science bibliography, https://dblp.org}
}

@inproceedings{WangDCK16,
  author       = {Yue Wang and
                  Neil T. Dantam and
                  Swarat Chaudhuri and
                  Lydia E. Kavraki},
  editor       = {Amanda Jane Coles and
                  Andrew Coles and
                  Stefan Edelkamp and
                  Daniele Magazzeni and
                  Scott Sanner},
  title        = {Task and Motion Policy Synthesis as Liveness Games},
  booktitle    = {Proceedings of the Twenty-Sixth International Conference on Automated
                  Planning and Scheduling, {ICAPS} 2016, London, UK, June 12-17, 2016},
  pages        = {536},
  publisher    = {{AAAI} Press},
  year         = {2016},
  url          = {http://www.aaai.org/ocs/index.php/ICAPS/ICAPS16/paper/view/13146},
  bibsource    = {dblp computer science bibliography, https://dblp.org}
}

@inproceedings{BowlingV03,
  author       = {Michael H. Bowling and
                  Manuela M. Veloso},
  editor       = {Georg Gottlob and
                  Toby Walsh},
  title        = {Simultaneous Adversarial Multi-Robot Learning},
  booktitle    = {IJCAI-03, Proceedings of the Eighteenth International Joint Conference
                  on Artificial Intelligence, Acapulco, Mexico, August 9-15, 2003},
  pages        = {699--704},
  publisher    = {Morgan Kaufmann},
  year         = {2003},
  url          = {http://ijcai.org/Proceedings/03/Papers/102.pdf},
  bibsource    = {dblp computer science bibliography, https://dblp.org}
}

@inproceedings{XiongL16,
  author       = {Liping Xiong and
                  Yongmei Liu},
  editor       = {Subbarao Kambhampati},
  title        = {Strategy Representation and Reasoning for Incomplete Information Concurrent
                  Games in the Situation Calculus},
  booktitle    = {Proceedings of the Twenty-Fifth International Joint Conference on
                  Artificial Intelligence, {IJCAI} 2016, New York, NY, USA, 9-15 July
                  2016},
  pages        = {1322--1329},
  publisher    = {{IJCAI/AAAI} Press},
  year         = {2016},
  url          = {http://www.ijcai.org/Abstract/16/191},
  bibsource    = {dblp computer science bibliography, https://dblp.org}
}

@misc{anonymous_2025_17357029,
  author       = {Anonymous},
  title        = {Artefact for TACAS 2026},
  month        = oct,
  year         = 2025,
  publisher    = {Zenodo},
  doi          = {10.5281/zenodo.17357028},
  url          = {https://doi.org/10.5281/zenodo.17357028},
}

@inproceedings{EmersonJutla91,
	Author = {E. A. Emerson and C. S. Jutla},
	booktitle = {FOCS'91},
	Pages = {368-377},
	Title = {Tree automata, mu-calculus and determinacy},
	Year = {1991}
}

@inproceedings{pnueli1989synthesis,
  title={On the synthesis of a reactive module},
  author={Pnueli, Amir and Rosner, Roni},
  booktitle={Proceedings of the 16th ACM SIGPLAN-SIGACT symposium on Principles of programming languages},
  pages={179--190},
  year={1989}
}

@Book{TabuadaBook,
  Title                    = {Verification and control of hybrid systems: a symbolic approach},
  Author                   = {Tabuada, P.},
  Publisher                = {Springer}, 
  Year                     = {2009}
}

@article{kress2018synthesisForRobotsReview,
	title        = {Synthesis for robots: Guarantees and feedback for robot behavior},
	author       = {Kress-Gazit, Hadas and Lahijanian, Morteza and Raman, Vasumathi},
	year         = 2018,
	journal      = {Annual Review of Control, Robotics, and Autonomous Systems},
	publisher    = {Annual Reviews},
	volume       = 1,
	number       = 1,
	pages        = {211--236}
}

@article{YIN2024100940,
title = {Formal synthesis of controllers for safety-critical autonomous systems: Developments and challenges},
journal = {Annual Reviews in Control},
volume = {57},
pages = {100940},
year = {2024},
issn = {1367-5788},
doi = {https://doi.org/10.1016/j.arcontrol.2024.100940},
author = {Xiang Yin and Bingzhao Gao and Xiao Yu},

}

@incollection{finkbeiner2016synthesis,
  title={Synthesis of reactive systems},
  author={Finkbeiner, Bernd},
  booktitle={Dependable Software Systems Engineering},
  pages={72--98},
  year={2016},
  publisher={IOS Press}
}

@inproceedings{HenzingerBook,
author = {Henzinger, Thomas A.},
title = {Games in system design and verification},
year = {2005},
isbn = {9810534124},
publisher = {National University of Singapore},
address = {SGP},
booktitle = {Proceedings of the 10th Conference on Theoretical Aspects of Rationality and Knowledge},
pages = {1–4},
numpages = {4},
location = {Singapore},
series = {TARK '05}
}

@inproceedings{ASN_HSCC24,
author = {Anand, Ashwani and Schmuck, Anne-Kathrin and Prakash Nayak, Satya},
title = {Contract-Based Distributed Logical Controller Synthesis},
year = {2024},
isbn = {9798400705229},
publisher = {Association for Computing Machinery},
address = {New York, NY, USA},
url = {https://doi.org/10.1145/3641513.3650123},
doi = {10.1145/3641513.3650123},
booktitle = {Proceedings of the 27th ACM International Conference on Hybrid Systems: Computation and Control},
articleno = {11},
numpages = {11},
location = {Hong Kong SAR, China},
series = {HSCC '24}
}

@article{kozen1983results,
  title={Results on the propositional $\mu$-calculus},
  author={Kozen, Dexter},
  journal={Theoretical computer science},
  volume={27},
  number={3},
  pages={333--354},
  year={1983},
  publisher={Elsevier}
}

@InProceedings{NS_TACAS24,
author="Nayak, Satya Prakash
and Schmuck, Anne-Kathrin",
editor="Finkbeiner, Bernd
and Kov{\'a}cs, Laura",
title="Most General Winning Secure Equilibria Synthesis in Graph Games",
booktitle="Tools and Algorithms for the Construction and Analysis of Systems",
year="2024",
publisher="Springer Nature Switzerland",
address="Cham",
pages="173--193",
isbn="978-3-031-57256-2"
}

@article{NNS_HSCC24,
  title={Context-triggered abstraction-based control design},
  author={Nayak, Satya Prakash and Egidio, Lucas N and Della Rossa, Matteo and Schmuck, Anne-Kathrin and Jungers, Raphael M},
  journal={IEEE Open Journal of Control Systems},
  volume={2},
  pages={277--296},
  year={2023},
  publisher={IEEE}
}

@article{JacobsPABCCDDDFFKKLMMPR24,
  author       = {Swen Jacobs and
                  Guillermo A. P{\'{e}}rez and
                  Remco Abraham and
                  V{\'{e}}ronique Bruy{\`{e}}re and
                  Micha{\"{e}}l Cadilhac and
                  Maximilien Colange and
                  Charly Delfosse and
                  Tom van Dijk and
                  Alexandre Duret{-}Lutz and
                  Peter Faymonville and
                  Bernd Finkbeiner and
                  Ayrat Khalimov and
                  Felix Klein and
                  Michael Luttenberger and
                  Klara J. Meyer and
                  Thibaud Michaud and
                  Adrien Pommellet and
                  Florian Renkin and
                  Philipp Schlehuber{-}Caissier and
                  Mouhammad Sakr and
                  Salomon Sickert and
                  Ga{\"{e}}tan Staquet and
                  Cl{\'{e}}ment Tamines and
                  Leander Tentrup and
                  Adam Walker},
  title        = {The Reactive Synthesis Competition {(SYNTCOMP):} 2018-2021},
  journal      = {Int. J. Softw. Tools Technol. Transf.},
  volume       = {26},
  number       = {5},
  pages        = {551--567},
  year         = {2024},
  url          = {https://doi.org/10.1007/s10009-024-00754-1},
  doi          = {10.1007/S10009-024-00754-1},
  bibsource    = {dblp computer science bibliography, https://dblp.org}
}

@book{NRTV2007,
  editor    = {Noam Nisan and Tim Roughgarden and Éva Tardos and Vijay V. Vazirani},
  title     = {Algorithmic Game Theory},
  publisher = {Cambridge University Press},
  year      = {2007},
  doi       = {10.1017/CBO9780511800481},
  url       = {https://doi.org/10.1017/CBO9780511800481}
}

\newpage
\appendix
\section{Algorithms}\label{app:algorithms}\label{app:algorithms_buchi}
In this section, we provide the full algorithm for computing the strategy templates for \buchi game in \cref{alg:buchiTemp}.

If $(\gamegraph, \LTLalways \LTLeventually \targetSet)$ is a concurrent \buchi game, then the algorithm $\buchiTemp$ (shown in \cref{alg:buchiTemp}) computes a safety function $\funcSafe$ and a liveness tuple $(\funcLive, \livePartitions)$ such that the template $\template = \templatesafe(\funcSafe) \cap \templatelive(\funcLive, \livePartitions)$ is winning for $\p{1}$ in the game. 
It first computes the winning region $\wino$ using the $\BuchiAlgo{\gamegraph, \targetSet}$, then it computes the safety function $\funcSafe$ using the procedure $\safeTemp(\gamegraph, \wino)$. Afterwards, it restricts the target set to $\targetSet \cap \wino$ and computes the liveness tuple $(\funcLive, \livePartitions)$ using the procedure $\liveTemp(\gamegraph, \targetSet)$.
Finally, we assign all actions to be live actions for every state outside of the winning region $\wino$, including some states of $\targetSet$, to ensure that the live function $\funcLive$ is well-defined.
Therefore, for every state outside of the winning region $\wino$, and every state in $\targetSet \cap \wino$ it assigns all actions to be live actions in $\funcLive$ (line \ref{line:lg:buchiTemp:liveoutside}). 
\begin{algorithm}[ht]
		\caption{$ \buchiTemp(\gamegraph, \targetSet) $}\label{alg:buchiTemp}
		\begin{algorithmic}[1]
			\Require A game graph $\gamegraph$, and a subset of states $ \targetSet$
			\Ensure A safety function $ \funcSafe $ and a liveness tuple $ (\funcLive, \livePartitions) $
			\State $\wino \gets \BuchiAlgo{\gamegraph,\targetSet}$; $ \funcSafe\gets \safeTemp(\gamegraph,\wino) $;
			\State $\targetSet \gets \targetSet\cap \wino$;%
			\State $ (\funcLive, \livePartitions)\gets \liveTemp(\gamegraph, \targetSet) $;
			\For{$ v \in I \cup (V\setminus \wino)$} $\funcLive(v) \gets \{\actiono(v)\setminus \funcSafe(v)\}$; \EndFor\label{line:lg:buchiTemp:liveoutside}
			\State \Return $ (\funcSafe, (\funcLive, \livePartitions)) $
			\Procedure{\liveTemp}{$ \gamegraph, \targetSet $}
				\State $ \livePartitions\gets\emptyset $;
				\While{$I \neq \wino$}
					\State $ I' \gets \Apre{1}{\wino}{I}$; $\livePartitions \gets \livePartitions \cup \{(I' \setminus I)\}$;
					\For{$ v \in I' \setminus I$}
						\State $ \funcLive(v) \gets \emptyset $;
						\For{$ \actt \in \actiont(v)$}
							\State $ \funcLive(v) \gets \funcLive(v) \cup \{\{\acto \in \actiono(v) \setminus \funcSafe(v) \mid \transition(v,\acto,\actt) \in I\}\}$;\label{line:lg:buchiTemp:removeunsafe}
						\EndFor
					\EndFor	
				\State $I \gets I'$;
				\EndWhile
				\State \Return $ (\funcLive, \livePartitions) $
			\EndProcedure
		\end{algorithmic}
	\end{algorithm}

\section{Formal Proofs of the Results}\label{app:proofs}
In this section, we restate every result from the main body of the paper and provide its formal proof.

\restatesafety*
\begin{proof}
	Let $\strato$ be a strategy that follows $\templatesafe(\funcSafe)$. For the rest of the proof, we denote $\wino(\LTLalways \targetSet)$ by $\wino$ for simplicity.
	We need to show that for every $v\in\wino$, $\strato$ is (almost surely) winning from $v$.
	However, we will show a stronger property, namely that for every $v\in\wino$, every $\strato$-play from $v$ belongs to $\lang(\Box \targetSet)$.
	Furthermore, it is easy to see that a strategy cannot win from a state outside of the safe set $\targetSet$. Therefore, $\wino\subseteq \targetSet$, and hence, it suffices to show that every $\strato$-play $\play\in \vertex^\omega$ from $\play_0 = v$ is contained in the winning region $\wino$, i.e., $\play_i \in \wino$ for all $i\geq 0$.
	We use induction to prove this argument.

	For base case, we have $\play_0\in\wino$ by construction.
	For induction case, we show that if $\play_i \in \wino$, then $\play_{i+1} \in \wino$.
	Suppose $\play_i \in \wino$, then according to \eqref{equ:SafeTempCompute}, we have:
	\[
	\funcSafe(\play_i) = \actiono(\play_i) \setminus \safeactions{\wino}{\play_i}{\emptyset}
	\]
	And, since $\strato$ follows $\templatesafe(\funcSafe)$, from \eqref{equ:PiSet} we have:
	\[
	\support(\strato(\play_{\leq i})) \cap \left(\actiono(\play_i) \setminus \safeactions{\wino}{\play_i}{\emptyset}\right) = \emptyset
	\]
	As $\support(\strato(\play_{\leq i})) \subseteq \actiono(\play_i)$, we have:
	\[
	\support(\strato(\play_{\leq i})) \subseteq \safeactions{\wino}{\play_i}{\emptyset}.
	\]
	Hence, by \eqref{equ:defAset}, it holds that for every $\acto \in \support(\strato(\play_{\leq i}))$ and every $\actt \in \actiont(\play_i)$, we have $\transition(\play_i,\acto,\actt) \in \wino$.
	As $\play$ is a $\strato$-play, by definition, $\play_{i+1} = \transition(\play_i,\acto,\actt)$ for some $\acto \in \support(\strato(\play_{\leq i}))$ and $\actt \in \actiont(\play_i)$.
	So, it must hold that $\play_{i+1} \in \wino$.
	Hence, by induction, we have shown that every $\strato$-play from $v$ is contained in $\lang(\Box \targetSet)$.
	\par
	Suppose $\strat$ is a winning strategy, then by definition, for every $\play\in \plays(\strat)$ and for every $i\geq 0$, every action $\acto \in \support(\strat(\play_{\leq i}))$ should lead to a state in $\wino$. That means, given a play $\play \in \plays(\strat)$, we have the following for every $i\geq 0$:
	\[
	\forall \acto \in \support(\strat(\play_{\leq i})) \cdot \forall \actt \in \actiont(\play_i) \cdot \transition(\play_i, \acto, \actt) \in \wino.
	\]
	According to \eqref{equ:defAset}, the above equation can be rewritten as:
	\[
	\forall \acto \in \support(\strat(\play_{\leq i})) \cdot \acto \in A^v_{\wino}(\emptyset).
	\]
	This implies the following:
	\begin{align*}
		&\support(\strat(\play_{\leq i})) \subseteq A^v_{\wino}(\emptyset)\\
			\Rightarrow & \support(\strat(\play_{\leq i}))\cap (\actiono(v) \setminus A^v_{\wino}(\emptyset)) = \emptyset\\
			\Rightarrow & \support(\strat(\play_{\leq i}))\cap \funcSafe(v) = \emptyset.
		\end{align*}
	Then, by definition of safety template $\templatesafe(\funcSafe)$, it holds that $\strat \in \templatesafe(\funcSafe)$.
	Hence, every \emph{almost sure} winning strategy $\strat$ in \emph{concurrent safety game} $\game=(\gamegraph, \Box \targetSet)$ follows the $\templatesafe(\funcSafe)$ proposed in \eqref{equ:SafeTempCompute}.
\end{proof}

\restatebuchi*
\begin{proof}
	Let $\strato$ be a strategy for $\p{1}$ that follows the template $\template = \templatesafe(\funcSafe) \cap \templatelive(\funcLive, \livePartitions)$. We need to show that $\strato$ is an almost surely winning strategy for $\p{1}$ in the game $\game=(\gamegraph, \LTLalways \LTLeventually \targetSet)$., i.e., $\probability_v^{\strato}(\lang(\LTLalways\LTLeventually \targetSet)) = 1$ for every $v \in \wino$.
	First, by the property of $\templatesafe(\funcSafe)$, every $\strato$-play $\play$ from any state $v \in \wino$ will remain in the winning region $\wino$, i.e., $\play \in \wino^\omega$.

	Now, let $\BuchiAlgo{\gamegraph, \targetSet}$ gives the sets of states $X_0\subseteq X_1 \subseteq \ldots \subseteq X_k = X_{k+1} = \wino$, where $X_1=\targetSet \cap \pre{1}{\wino}$ and $X_{i+1} = (\neg \targetSet \cap \Apre{1}{\wino}{X_i}) \cup X_1$ for $i \ge 1$.
	Let $\livePartition_i\in\livePartitions$ be the $i$-th element of the set $\livePartitions$, i.e., $\livePartition_i = X_i \setminus X_{i-1}$, and $\livePartition_1 = I$.
	Let $E_i$ be the set of plays that visits a state in $X_i$ infinitely many times, i.e., $E_i = \{\play \mid X_i\cap \infPlay{\play} \ne \emptyset\}$.
	As every $\strato$-play $\play$ is an infinite sequence of states, it visits $X_k = \wino$ infinitely often.
	Hence, $\probability_v^{\strato}(E_k) = 1$ for every $v \in \wino$.
	Now, we will inductively show that if for some $i \ge 2$, we have $\probability_v^{\strato}(E_i) = 1$ for all $v\in \wino$, then it also holds that $\probability_v^{\strato}(E_{i-1}) = 1$ for all $v\in \wino$.

	Suppose for some $i \ge 2$, $\probability_v^{\strato}(E_i) = 1$ for all $v\in \wino$ but $\probability_v^{\strato}(E_{i-1}) \ne 1$ for some $v\in \wino$.
	Consider the $\strato$-plays $\play$ that belongs to $E_i$ but not to $E_{i-1}$.
	Hence, $\play$ visits $X_i$ infinitely often but visits $X_{i-1}$ finitely often, that means, $\play$ must visit $\livePartition_i = X_i \setminus X_{i-1}$ infinitely often.
	As $\strato$ follows $\templatelive(\funcLive, \livePartitions)$, according to \eqref{equ:PiSetLive}, $\livePartition_i \cap \infPlay{\play} \ne \emptyset$ implies:
	\[
	\sum_{\{j \mid \play_j\in \livePartition_i\}} \minsupport{\strato(\play_{\leq j} )}{\funcLive(\play_j)} = \infty.
	\]
	From the defintion of $\minsupp$ \eqref{equ:MinSupp}, we can write the above equation as:
	\[
	\sum\limits_{\{j \mid \play_j\in \livePartition_i\}} \left[\min\limits_{\subactiono \in \funcLive(\play_j)} \left\{\sum\limits_{\acto \in \subactiono} \strato(\play_{\leq j})(\acto)\right\}\right] = \infty
	\]
	When summation over minimum value of a set is infinite, it means that the summation over every element of the set is infinite. Furthermore, if we name $\play_j = v'$, in $\buchiTemp$, for every $\actt \in \actiont(v')$, $\{\acto \in \actiono(v') \setminus \funcSafe(v') \mid \transition(v',\acto, \actt) \in I\}$ has been added to $\funcLive(v')$. Hence:
	\[
	\exists v' \in \livePartition_i \cdot \forall \actt \in \actiont(v') \cdot \exists \subactiono \subseteq \actiono(v') \cdot \left[\transition(v', \subactiono, \actt) \subseteq \livePartition_{i-1} \land \sum\limits_{\{j \mid \play_j=v'\}}  \left\{\sum\limits_{\acto \in \subactiono} \strato(\play_{<j})(\acto)\right\}=\infty\right] 
	\]
	This means, summation over the probabilities of actions reaching $\livePartition_{i-1}$ in visits to $v'$ is infinite. 
	So, we conclude that, in the sample space of $\strato$-plays that visits $\livePartition_i$ infinitely often, the probability of visiting $\livePartition_{i-1}$ infinitely often is $1$.\\
	If the probability of visiting $\livePartition_{i-1}$ from $v'$ in $j$-th visit to $v'$ is $p_j$, then the probability of not visiting $\livePartition_{i-1}$ from $v'$ in $j$-th visit to $v'$ is $(1-p_j)$. By the convergence criterion for infinite products, since $\forall j \cdot 0<p_j<1$ and $\sum\limits_{j} p_j$ diverges, then $\prod\limits_{j} (1-p_j)$ converges to $0$. Hence, the probability of visiting $\livePartition_{i-1}$ infinitely often is $1$.	
	Hence, it contradicts our assumption that $\probability_v^{\strato}(E_i) = 1$ and $\probability_v^{\strato}(E_{i-1}) \ne 1$.

	With the above inductive argument, we have shown that for all $i\geq 1$ $\probability_v^{\strato}(E_i) = 1$ for all $v\in \wino$. Hence, $\probability_v^{\strato}(E_1) = 1$ for all $v\in \wino$, which means $\probability_v^{\strato}(\lang(\LTLalways\LTLeventually \targetSet)) = 1$ for all $v\in \wino$.
	Therefore, $\strato$ is an almost surely winning strategy for $\p{1}$ in the game $\game=(\gamegraph, \LTLalways\LTLeventually \targetSet)$.
\end{proof}

\restatecobuchi*
\begin{proof}
	Let $\strato$ be a strategy for $\p{1}$ that follows the template $\template = \templatesafe(\funcSafe) \cap \templatelive(\funcLive, \livePartitions) \cap \templatecolive(\funcCoLive)$. We need to show that $\strato$ is an almost surely winning strategy for $\p{1}$ in the game $\game=(\gamegraph, \LTLeventually\LTLalways \targetSet)$., i.e., $\probability_v^{\strato}(\lang(\LTLeventually\LTLalways \targetSet)) = 1$ for every $v \in \wino$.
	First, by the property of $\templatesafe(\funcSafe)$, every $\strato$-play $\play$ from any state $v \in \wino$ will remain in the winning region $\wino$, i.e., $\play \in \wino^\omega$.
	As every $\strato$-play $\play$ is an infinite sequence of states, it visits $\wino$ infinitely often. Hence, $\infPlay{\play} \cap \wino \ne \emptyset$.\\
	Now, let $\CobuchiAlgo{\gamegraph, \targetSet}$ gives the sets of states $X_0\subseteq X_1 \subseteq \ldots \subseteq X_k = X_{k+1} = Z^* = \wino$, where $X_1= \wino(\LTLalways \targetSet)$.
	Let $\livePartition_i\in\livePartitions$ be the $i$-th element of the set $\livePartitions$, i.e., $\livePartition_i = X_i \setminus X_{i-1}$, and $\livePartition_1 = X_1$.
	Let $E_i$ be the set of plays that visits a state in $X_i$ infinitely many times, i.e., $E_i = \{\play \mid X_i\cap \infPlay{\play} \ne \emptyset\}$.
	Now, we will show that if we have $\probability_v^{\strato}(E_1) = 1$ for all $v\in \wino$, then $\probability_v^{\strato}({E_1}^\complement) = 0$ for all $v\in \wino$, where ${E_1}^\complement$ is the complement of $E_1$. Intuitively, this means that if probability of visiting $X_1$ infinitely often is $1$, then the probability of visiting ${X_1}^\complement$ infinitely often is $0$.

	As $\strato$ follows $\templatecolive(\funcCoLive)$, according to \eqref{equ:PiSetCoLive2}, for every $v \in X_1$ we have:
	\[
	 v \in \infPlay{\play} \implies \sum\limits_{\{i \mid \play_i=v\}} \strato(\play_{\leq i})(\funcCoLive(\play_i)) \ne \infty	
	\]
	In the \cref{alg:cobuechi_temp}, we have used $\safeTemp(\gamegraph, X_1)$ to define $\funcCoLive(v)$ for every $v \in X_1$. Hence, $\funcCoLive(v)$ contains all the actions that can lead outside $X_1$. Therefore, by the \emph{Borel-Cantelli} lemma, the above equation implies that, for every $v \in X_1$, the actions that can lead outside $X_1$ are chosen only finitely many times in $\play$. This means that, if $\probability_v^{\strato}(E_1) = 1$ for all $v\in \wino$, then $\probability_v^{\strato}(E_1^\complement) = 0$.\\
	Now, we will show that $\probability_v^{\strato}(\lang(\LTLalways\LTLeventually \neg \targetSet)) = 0$ for all $v\in \wino$, i.e. the probability of visiting $\neg \targetSet$ infinitely often is $0$. We use contradiction for this purpose. Suppose, there exists a state $v' \in \wino \cap \neg \targetSet$ such that $v' \in \infPlay{\play}$ with positive probability.\\
	Firstly note that, $X_1 \cap \neg \targetSet = \emptyset$, because $X_1 = \wino(\LTLalways \targetSet)$. Hence, $v' \notin X_1$ and $v' \in X_i$ for some $i \ge 2$. Since $\strato$ follows $\templatelive(\funcLive, \livePartitions)$, according to \cref{thm:buchi}, we know if $\probability_v^{\strato}(E_i) = 1$, then $\probability_v^{\strato}(E_{i-1}) = 1$ for all $v\in \wino$ and $i \ge 2$.
	By using the same inductive argument as in \cref{thm:buchi}, we can show that $\probability_v^{\strato}(E_{1}) = 1$ for all $v\in \wino$.
	But, we showed that if $\probability_v^{\strato}(E_1) = 1$ for all $v\in \wino$, then $\probability_v^{\strato}(E_1^\complement) = 0$. Hence, we have a contradiction, because we assumed that $v' \in \infPlay{\play} \cap {E_1}^\complement$ with positive probability.\\
	Now, by negation of the logical formula, we have $\probability_v^{\strato}({\lang(\LTLeventually\LTLalways  \targetSet)}^\complement) = 0$, for all $v\in \wino$, and this is equivalent to $\probability_v^{\strato}(\lang(\LTLeventually\LTLalways  \targetSet)) = 1$ for all $v\in \wino$. Therefore, $\strato$ is an almost surely winning strategy for $\p{1}$ in the game $\game=(\gamegraph, \LTLeventually\LTLalways \targetSet)$.
		
\end{proof}

\restateallconflictfree*
\begin{proof}
We will prove the conflict-freeness of the templates computed by $\safeTemp$, $\buchiTemp$, and $\cobuchiTemp$.

\noindent\textbf{For safety games:} Suppose $\funcSafe = \safeTemp(\gamegraph, \targetSet)$, then by construction, the template $\templatesafe(\funcSafe)$ does not restrict any action from states outside the winning region $\wino$. Furthermore, by the proof of \cref{thm:safety}, we know that for every state $v \in \wino$ and for every $\play \in \plays(\strat)$ with $\play_i = v$, there exists a winning strategy $\strat$ for $\p{1}$ from $v$ such that $\support(\strat(\play_{\leq i})) \cap \funcSafe(v) = \emptyset$. Hence, $\templatesafe(\funcSafe)$ is conflict-free.

\noindent\textbf{For \buchi games:} Suppose $(\funcSafe, (\funcLive, \livePartitions)) = \buchiTemp(\gamegraph, \targetSet)$, then by construction again, the template $\template = \templatesafe(\funcSafe) \cap \templatelive(\funcLive, \livePartitions)$ does not restrict any action from states outside the winning region $\wino$. Furthermore, by \cref{line:lg:buchiTemp:removeunsafe} of \cref{alg:buchiTemp}, for every state $v \in \wino$, $\funcLive(v)$ does not contain any unsafe actions from $\funcSafe(v)$. Hence, $\template$ is conflict-free.

\noindent\textbf{For co-\buchi games:} Suppose $(\funcSafe, (\funcLive, \livePartitions), \funcCoLive) = \cobuchiTemp(\gamegraph, \targetSet)$, then by construction again, the template $\template = \templatesafe(\funcSafe) \cap \templatelive(\funcLive, \livePartitions) \cap \templatecolive(\funcCoLive)$ does not restrict any action from states outside the winning region $\wino$.
Furthermore, by \cref{line:alg:cobuechi_temp:removeunsafe} of \cref{alg:cobuechi_temp}, for every state $v \in \wino$, $\funcCoLive(v)$ does not contain any unsafe actions from $\funcSafe(v)$. 
Also, by \cref{line:alg:cobuechi_temp:safecolive} of \cref{alg:cobuechi_temp}, colive actions only contain actions from $X$, and by \cref{line:alg:cobuechi_temp:liveoutside} of \cref{alg:cobuechi_temp}, for every state $v \in X$, $\funcLive(v)$ contains all the possible actions from $v$.
Due to conflict-freeness of $\safeTemp$, there will always exist available actions in $\funcLive(v)$ from $X$ that are not in $\funcCoLive(v)$.
Hence, $\template$ is conflict-free.
	
\end{proof}

\section{Additional Proofs of Completeness of Algorithms}

\begin{restatable}{theorem}{restatebuchiComplete}\label{thm:buchiComplete}
	Given the premises of \cref{thm:buchi}, the template $\template = \templatesafe(\funcSafe) \cap \templatelive(\funcLive, \livePartitions)$ is complete, meaning that if there exists a winning strategy for $\p{1}$ in the game $\game=(\gamegraph,\LTLalways \LTLeventually \targetSet)$, then there exists a strategy that follows the template $\template$.
\end{restatable}
\begin{proof}
	Suppose $\strat$ is a winning strategy for $\p{1}$ in the game $\game=(\gamegraph, \LTLalways\LTLeventually \targetSet)$. 
	First of all, since $\strat$ is a winning strategy for $\p{1}$, it should never leave the winning region $\wino$, i.e. $\strat$ should never choose an action that leads to a state outside of $\wino$. As we proved in \cref{thm:safety}, the template $\templatesafe(\funcSafe)$ is non-empty.

	The set of all states that from which there is a strategy to reach $\targetSet$ in exactly one step (with probability 1), is $\Apre{1}{\wino}{\targetSet}$, which we defined it as $\livePartition_1$ in the procedure $\liveTemp(\gamegraph, \targetSet)$.
	If a $\strat$-play $\play$ visited $\livePartition_1$ infinitely often, there must be a state $v \in \livePartition_1$ such that, $\strat$ is a winning strategy for $\p{1}$ at $v$, i.e. $\strat$ should assign positive probability to some action against every possible action of $\p{2}$, to be able to reach $\livePartition_{0} = \targetSet$ infinitely many times.
	If it was not the case, then $\p{2}$ could choose an action that prevents $\p{1}$ from reaching $\targetSet$ from $v$, which contradicts the fact that $\strat$ is a winning strategy for $\p{1}$.
	We can write this fact as below:
	\begin{eqnarray*}
		\livePartition_1 \cap \infPlay{\play} \ne \emptyset \rightarrow \exists v \in \livePartition_1 \cdot \forall b \in \actiont(v) \cdot \exists \subactiono \subseteq \actiono(v) \cdot \\
		\left[\transition(v, \subactiono, b) \subseteq \livePartition_{0} \land \sum\limits_{\{j \mid \play_j=v\}} \left\{\sum\limits_{\acto \in \subactiono} \strat(\play_{<j})(\acto)\right\} = \infty\right]		
	\end{eqnarray*}
	Since in the $\buchiTemp(\gamegraph, \targetSet)$, for every $\actt \in \actiont(v)$, we are adding the set $\subactiono^{\actt} = \{\acto \in \actiono(v) \setminus \funcSafe(v)\mid \transition(v,\acto,\actt) \in \livePartition_{0}\}$ to $\funcLive(v)$, as defined in \cref{sec:buechi}, we can conclude that $\subactiono$ is a subset of $\subactiono^{\actt}$, because $\subactiono^{\actt}$ has all the possible \emph{safe} actions reaching $\livePartition_0$.
	So, always $\sum\limits_{\acto \in \subactiono} \strat(\play_{<j})(\acto)$ is less than $\sum\limits_{\acto \in \subactiono^{\actt}} \strat(\play_{<j})(\acto)$ for all $j$. Therefore, if summation over the first one is infinite, then summation over the second one is also infinite. As a result, we can change the above equation as below:
	\begin{eqnarray*}
		\livePartition_1 \cap \infPlay{\play} \ne \emptyset \rightarrow \exists v \in \livePartition_1 \cdot \forall b \in \actiont(v) \cdot \exists \subactiono^{\actt} \subseteq \actiono(v) \cdot \\
		\left[\transition(v, \subactiono^{\actt}, b) \subseteq \livePartition_{0} \land \sum\limits_{\{j \mid \play_j=v\}} \left\{\sum\limits_{\acto \in \subactiono^{\actt}} \strat(\play_{<j})(\acto)\right\} = \infty\right]
	\end{eqnarray*}
	Since summation over $\subactiono^{\actt}$ for every $b \in \actiont(v)$ is infinite, the \emph{minimum} of these summations is also infinite. So, if we fix the definition of $\subactiono^{\actt}$ as in \cref{sec:buechi}, we'll have:
	\[
	\livePartition_1 \cap \infPlay{\play} \ne \emptyset \rightarrow \exists v \in \livePartition_1 \cdot \min\limits_{b \in \actiont(v)} \left[\sum\limits_{\{j \mid \play_j=v\}} \left\{\sum\limits_{\acto \in \subactiono^{b}} \strat(\play_{<j})(\acto)\right\}\right] = \infty
	\]
	Here variable of \emph{minimum} is independent of first summation, so they can be swapped:
	\[
	\livePartition_1 \cap \infPlay{\play} \ne \emptyset \rightarrow \exists v \in \livePartition_1 \cdot \sum\limits_{\{j \mid \play_j=v\}} \left[\min\limits_{b \in \actiont(v)} \left\{\sum\limits_{\acto \in \subactiono^{b}} \strat(\play_{<j})(\acto)\right\}\right] = \infty
	\]
	From the definition of \eqref{equ:MinSupp}, we can write the above equation as:
	\[
	\livePartition_1 \cap \infPlay{\play} \ne \emptyset \rightarrow \exists v \in \livePartition_1 \cdot\left[ \sum\limits_{\{j \mid \play_j=v\}} \minsupport{\strat(\play_{<j})}{\{\subactiono^b \mid \forall \actt \in \actiont(v)\}}\right] = \infty
	\]
	As mentioned before, $\{\subactiono^b \mid \forall \actt \in \actiont(v)\}$ is what we have defined as $\funcLive(v)$ in the procedure $\liveTemp(\gamegraph, \targetSet)$. So, we can write:
	\[
	\livePartition_1 \cap \infPlay{\play} \ne \emptyset \rightarrow \exists v \in \livePartition_1 \cdot\left[ \sum\limits_{\{j \mid \play_j=v\}} \minsupport{\strat(\play_{<j})}{\funcLive(v)}\right] = \infty
	\]
	Then, we can conclude the summation over all states in $\livePartition_1$ would be infinite, because one of them is infinite:
	\[
	\livePartition_1 \cap \infPlay{\play} \ne \emptyset \rightarrow \sum\limits_{v \in \livePartition_1}\left[ \sum\limits_{\{j \mid \play_j=v\}} \minsupport{\strat(\play_{<j})}{\funcLive(v)}\right] = \infty
	\]
	Finally, according to \eqref{equ:PiSetLive}, we have shown that, for this partition of $\wino$, the template $\templatelive(\funcLive, \livePartitions)$ is non-empty.
	Then, we can use the same argument for $\livePartition_2$, which is the set of states that can reach $\livePartition_1$ in exactly one step, and so on, until we conclude that for every $\livePartition_i$, $i \ge 1$, the template $\templatelive(\funcLive, \livePartitions)$ is non-empty.		\\
	We should just make sure that the $\template$ is non-empty for the states in $\livePartition_0$, which is a subset of goal set $\targetSet$.
	That is also guaranteed by the $\templatesafe(\funcSafe)$, as it was the only template we used to restrict the game graph $\gamegraph$ to the winning region $\wino(\Box \Diamond \targetSet)$, and we have proved in \cref{thm:safety} that the template $\templatesafe(\funcSafe)$ is non-empty for every state in $\wino$.\\
	Therefore, we have shown that if there exists a winning strategy for $\p{1}$ in the game $\game=(\gamegraph, \LTLalways \LTLeventually \targetSet)$, we can conclude that the template $\template = \templatesafe(\funcSafe) \cap \templatelive(\funcLive, \livePartitions)$ is non-empty for every state in $\wino$, including the states in $\livePartition_0$, i.e. the algorithm $\buchiTemp$ is complete.

\end{proof}

\begin{restatable}{theorem}{restatecobuchiComplete}\label{thm:cobuchiComplete}
		Given the premises of \cref{thm:cobuchi}, the template $\template = \templatesafe(\funcSafe) \cap \templatelive(\funcLive, \livePartitions) \cap \templatecolive(\funcCoLive)$ is complete, meaning that if there exists a winning strategy for $\p{1}$ in the game $\game=(\gamegraph,\LTLeventually\LTLalways \targetSet)$, then there exists a strategy that follows the template $\template$.
\end{restatable}
\begin{proof}
	
	Suppose $\strat$ is a winning strategy for $\p{1}$ in the game $\game=(\gamegraph, \LTLeventually\LTLalways \targetSet)$. 
	First of all, since $\strat$ is a winning strategy for $\p{1}$, it should never leave the winning region $\wino$, i.e. $\strat$ should never choose an action that leads to a state outside of $\wino$. As we proved in \cref{thm:safety}, the template $\templatesafe(\funcSafe)$ is non-empty for every state in $\wino$.\\
	Now, let $\CobuchiAlgo{\gamegraph, \targetSet}$ gives the sets of states $X_0\subseteq X_1 \subseteq \ldots \subseteq X_k = X_{k+1} = Z^* = \wino$, where $X_1= \wino(\LTLalways \targetSet)$, and define $\infStrat{\strat(v)}$ as below:
	\begin{equation}\label{equ:infStrat}
		\infStrat{\strat(v)} = \operatorname*{argmax}_{\{\subaction \subseteq \actiono(v) \mid \sum\limits_{\{i \mid v_i = v\}} \min\limits_{\acto \in \subaction} \strat(v_0 v_1 ... v_i)(\acto) = \infty\}}|\subaction|
	\end{equation}
	Intuitively, $\infStrat{\strat(v)}$ is the greatest set of actions that are in the support of the strategy $\strat$ infinitely many times, simultaneously, by the strategy $\strat$ in state $v$.\\
	Now, we will show that for every $v \in X_1$, where we know $X_1 = \wino(\LTLalways \targetSet)$, it holds that $\strat \in \templatecolive(\funcCoLive)$.\\
	For every $v \in X_1$, We defined $\funcCoLive(v)$ in the \cref{alg:cobuechi_temp} by using $\safeTemp(\gamegraph, X_1)$.
	Also, if every $\strat$-play $\play = v_0 v_1 v_2 \ldots$ visited $X_1$ infinitely many times, should leave $X_1$ finitely often, i.e. $\strat$ should assign zero probability to actions that can lead outside $X_1$, unsafe actions with respect to $X_1$, from some point on. Otherwise, $\play$ would leave $X_1$ infinitely many times, which contradicts the fact that $\strat$ is a winning strategy for $\p{1}$.
	We can write this fact as below:
	\begin{eqnarray*}
		&v \in \infPlay{\play} \implies \\
		&\forall \acto \in \actiono(v)\cdot \Bigl(\exists \actt \in \actiont(v) \cdot \transition(v,\acto,\actt) \notin X_1 \implies \sum\limits_{\{i \mid \play_i=v\}} \strat(\play_{\le i})(\acto) \neq \infty\Bigr)
	\end{eqnarray*}
	If we negate the left-hand side of the implication, we get:
	\begin{eqnarray*}
		&v \in \infPlay{\play} \implies 
		\\
		&\forall \acto \in \actiono(v)\cdot \Bigl(\neg (\forall \actt \in \actiont(v) \cdot \transition(v,\acto,\actt) \in X_1) \implies \sum\limits_{\{i \mid \play_i=v\}} \strat(\play_{\le i})(\acto) \neq \infty\Bigr)
	\end{eqnarray*}
	By the \cref{equ:defAset}, we can rewrite the above equation as below:
	\[
	v \in \infPlay{\play} \implies \forall \acto \in \actiono(v)\cdot \Bigl(\acto \notin \safeactions{X_1}{v}{\emptyset} \implies \sum\limits_{\{i \mid \play_i=v\}} \strat(\play_{\le i})(\acto) \neq \infty\Bigr)
	\]
	This is equivalent to:
	\[
	v \in \infPlay{\play} \implies \forall \acto \in \actiono(v) \setminus \safeactions{X_1}{v}{\emptyset}\cdot \Bigl(\sum\limits_{\{i \mid \play_i=v\}} \strat(\play_{\le i})(\acto) \neq \infty\Bigr)
	\]
	By the definition of $\safeTemp$ \eqref{equ:SafeTempCompute}, and since we defined $\funcCoLive(v)$ for every $v \in X_1$ by using $\safeTemp(\gamegraph, X_1)$ in the \cref{alg:cobuechi_temp}, we know that $\funcCoLive(v) = \actiono(v) \setminus \safeactions{X_1}{v}{\emptyset}$. Therefore, we can rewrite the above equation as below:
	\[
	v \in \infPlay{\play} \implies \forall \acto \in \funcCoLive(v)\cdot \Bigl(\sum\limits_{\{i \mid \play_i=v\}} \strat(\play_{\le i})(\acto) \neq \infty\Bigr)
	\]
	Since, the summation for every $\acto \in \funcCoLive(v)$ is not infinite, the summation over all the actions in $\funcCoLive(v)$ is not infinite, too. Hence, we can change the above equation to:
	\[
	v \in \infPlay{\play} \implies \Bigl(\sum\limits_{\{i \mid \play_i=v\}} \strat(\play_{\le i})(\funcCoLive(\play_i)) \neq \infty\Bigr)
	\]
	We showed that this argument holds for every $\strat$-play $\play \in \plays(\strat)$, and every state $v \in X_1$. Therefore, we can conclude that $\strat \in \templatecolive(\funcCoLive)$ \eqref{equ:PiSetCoLive2}.\\
	Now, we will show that for every $v \in X_i$ where $i \ge 2$, it holds that $\strat \in \template$. We know that every $X_i, i \ge 2$ can contain states from both $\targetSet$ and $\neg \targetSet$. So we will consider these two cases separately. First, start with the states in $\neg \targetSet$.\\
	In \cref{alg:cobuechi_temp}, we defined $\funcLive(v)$ for every $v \in X_i \cap \neg \targetSet$, for $i \ge 2$, by using $\Apre{1}{Z}{X_{i-1}}$, as for \buchi games. Therefore, by using the same argument as in \cref{thm:buchiComplete}, we can show that $\strat \in \templatelive(\funcLive, \livePartitions)$ for every $v \in X_i \cap \neg \targetSet$. But, for states $v \in X_i \cap \targetSet$, for $i \ge 2$, we defined $\funcLive(v)$ by using $\AFpre{1}{Z}{Y}{X_{i-1}}$. Hence, we are going to present new lemmas that will help us to show that $\strat \in \templatelive(\funcLive, \livePartitions)$ for every $v \in X_i \cap \targetSet$.\\
	\begin{lemma}\label{lem:cobuchi1}
	Let $\game=(\gamegraph, \LTLeventually \LTLalways\targetSet)$ be a \emph{concurrent \cobuchi game}, and the set of states $X_i$ and $X_{i-1}$ be given by $\CobuchiAlgo{\gamegraph, \targetSet}$, for some $i \ge 2$. If $\subaction \subseteq \safeactions{X_i}{v}{\forwardingactions{X_{i-1}}{v}{\subaction}}$ for every $v \in X_i \cap \targetSet$, then every $\strat$ with $\subaction = \infStrat{\strat(v)}$ is a winning strategy for $\p{1}$ in the game $\game$.
	\end{lemma}
	\begin{proof}
		We can write the if condition of the lemma as below:
		\[
		\acto^* \in \subaction \implies \acto^* \in \safeactions{X_i}{v}{(\forwardingactions{X_{i-1}}{v}{\subaction})}
		\]
		Then, by the \cref{equ:defAset}, we can rewrite the above equation as below:
		\[
		\acto^* \in \subaction \implies \forall \actt \in \actiont(v) \cdot \Bigl(\transition(v,\acto^*,\actt) \notin X_i \implies \actt \in \forwardingactions{X_{i-1}}{v}{\subaction}\Bigr)
		\]
		By negating the left-hand side of the implication, we get:
		\[
		\acto^* \in \subaction \implies \forall \actt \in \actiont(v) \cdot \Bigl(\transition(v,\acto^*,\actt) \in X_i \lor \actt \in \forwardingactions{X_{i-1}}{v}{\subaction}\Bigr)
		\]
		By the \cref{eq:defBset}, we can rewrite the above equation as below:
		\[
		\acto^* \in \subaction \implies \forall \actt \in \actiont(v) \cdot \Bigl(\transition(v,\acto^*,\actt) \in X_i \lor \left(\exists \acto \in \subaction \cdot \transition(v,\acto,\actt) \in X_{i-1}\right)\Bigr)
		\]		
		Intuitively, the above equation means that by giving positive probability to every action in $\subaction$, for every $v \in X_i \cap \targetSet$, there will be no action for the opponent to prevent the player to reach $X_{i-1}$ and transition to the outside of $X_i$. Hence, by giving positive probability to every action in $\subaction$, the strategy $\strat$ is winning in $v$.
		(By staying in $X_i$ infinitely many times, if some $v \in X_i \cap \neg \targetSet$ is visited infinitely many times, then by the same argument as in \cref{thm:buchi}, $X_{i-1}$ will be visited infinitely many times, too. Otherwise, it remains in stays in $X_i$ forever, and visits $\neg \targetSet$ finitely many times, and this is winning as well.)\\
	
	\end{proof}
	\begin{lemma}\label{lem:cobuchi2}
	Let $\game=(\gamegraph, \LTLeventually \LTLalways\targetSet)$ be a \emph{concurrent \cobuchi game}, and the set of states $X_i$ and $X_{i-1}$ be given by $\CobuchiAlgo{\gamegraph, \targetSet}$, for some $i \ge 2$. If $\forwardingactions{X_{i-1}}{v}{\subaction} = \actiont(v)$ for every $v \in X_i \cap \targetSet$, then every $\strat$ with $\subaction = \infStrat{\strat(v)}$ follows the template $\templatelive(\funcLive, \livePartitions)$.
	\end{lemma}
	\begin{proof}
		We can write the if condition of the lemma as below:
		\[
		\actt \in \actiont(v) \implies \Bigl(\actt \in \forwardingactions{X_{i-1}}{v}{\subaction}\Bigr)
		\]
		Then, by the \cref{eq:defBset}, we can rewrite the above equation as below:
		\[
		\actt \in \actiont(v) \implies \Bigl(\exists \acto \in \subaction \cdot \transition(v,\acto,\actt) \in X_{i-1}\Bigr)
		\]
		Since In \cref{alg:cobuechi_temp}, for every $\actt \in \actiont(v)$ we are adding $\{\acto \in \actiono(v) \mid \transition(v,\acto,\actt) \in X_{i-1}\}$ to $\funcLive(v)$, namely $h_{\actt}$, we can conclude the following:
		\[
		\actt \in \actiont(v) \implies \Bigl(\exists \acto \in \subaction \cdot \acto \in h_{\actt} \Bigr)
		\]
		By the definition of $\infStrat{\strat(v)}$ \eqref{equ:infStrat}, we know that $\sum\limits_{\{i \mid v_i = v\}} \min\limits_{\acto \in \subaction} \strat(v_0 v_1 ... v_i)(\acto) = \infty$. When the sum of minimum probabilities of actions in $\subaction$ is infinite, the sum of probabilities of actions in every non-empty subset of $\subaction$ is also infinite. Hence, by the above argument, we can conclude that the sum of probabilities of actions in $h_{\actt}$ is infinite, too. Therefore, with respect to \eqref{equ:MinSupp}, we can write:
		\[
		\sum\limits_{\{i \mid v_i = v\} } \minsupport{\strat(v_0 v_1 ... v_i)}{\{h_{\actt} \mid \actt \in \actiont\}} = \infty
		\] 
		Which is equivalent to:
		\[
		\sum\limits_{\{i \mid v_i = v\} } \minsupport{\strat(v_0 v_1 ... v_i)}{\funcLive(v)} = \infty
		\] 
		By the definition of $\templatelive(\funcLive, \livePartitions)$ \eqref{equ:PiSetLive}, we can conclude that $\strat \in \templatelive(\funcLive, \livePartitions)$.\\
		\end{proof}
	\begin{lemma}\label{lem:cobuchi3}
	Let $\game=(\gamegraph, \LTLeventually \LTLalways\targetSet)$ be a \emph{concurrent \cobuchi game} with a winning strategy $\strat$ with $\subaction = \infStrat{\strat(v)}$ for $\p{1}$, and the set of states $X_i$ and $X_{i-1}$ be given by $\CobuchiAlgo{\gamegraph, \targetSet}$, for some $i \ge 2$. If $\forwardingactions{X_{i-1}}{v}{\subaction} \ne \actiont(v)$, then there is another winning strategy $\strat'$ with $\subaction' = \infStrat{\strat'(v)}$ for $\p{1}$, such that $\forwardingactions{X_{i-1}}{v}{\subaction'} = \actiont(v)$, for every $v \in X_i \cap \targetSet$.
	\end{lemma}
	\begin{proof}
		Firstly, we define $\actiont'(v) := \actiont(v)\setminus \forwardingactions{X_{i-1}}{v}{\subaction}$ which is not empty, because of the premise of the lemma. Secondly, for every $\actt \notin \forwardingactions{X_{i-1}}{v}{\subaction}$, for all $\acto \in \subaction$, we have $\transition(v,\acto,\actt) \notin X_{i-1}$, by the \cref{eq:defBset}. But, $\transition(v, \acto,\actt) \in \targetSet$, because $\strat$ is a winning strategy, and if it's not going to $X_{i-1}$, it should visit $\targetSet$ infinitely many times. Otherwise, it will remain in $X_i$ forever, and visit $\neg \targetSet$ infinitely many times, which is losing. Therefore, we can write:
		\[
		\forall \actt \in \actiont'(v) \cdot \forall \acto \in \subaction \cdot \Bigl[\transition(v,\acto,\actt) \notin X_{i-1} \land \transition(v,\acto,\actt) \in \targetSet\Bigr]
		\]
		In the other hand, since $v \in \wino$, for every opponent's action, there should be at least one action for the player which goes to $X_{i-1}$:
		\[
		\forall \actt \in \actiont'(v) \cdot \exists \acto \in \actiono(v) \cdot \transition(v,\acto,\actt) \in X_{i-1}
		\]
		From the earlier argument, we know that $\acto \notin \subaction$. Therefore, we can change the above equation to:
		\[
		\forall \actt \in \actiont'(v) \cdot \exists \acto^* \in \actiono(v) \setminus \subaction \cdot \transition(v,\acto^*,\actt) \in X_{i-1}
		\]
		If we add all these $\acto^*$s to $\subaction$, we will compute a new set called $\subaction'$:
		\[
		\subaction' := \subaction \cup \{\acto^* \in \actiono(v) \setminus \subaction \mid \exists \actt \in \actiont'(v) \cdot \transition(v,\acto^*,\actt) \in X_{i-1}\}	
		\]
		Now, if strategy $\strat'$ gives positive probability to every action in $\subaction'$, for every $v \in X_i \cap \targetSet$, $\infPlay{\strat'(v)} = \subaction'$. Since for every $\actt \in \actiont(v)$ we added such $\acto^* \in \actiono(v)$ to $\subaction$ which $\transition(v,\acto^*,\actt) \in X_{i-1}$, and by the definition of \eqref{eq:defBset}, we can write:
		\[
		\forwardingactions{X_{i-1}}{v}{\subaction'} = \actiont(v)
		\]
		So, we found a new strategy $\strat'$ such that $\forwardingactions{X_{i-1}}{v}{\subaction'} = \actiont(v)$, for every $v \in X_i \cap \targetSet$.\\
		For the rest of the proof, we should show that $\strat'$ is a winning strategy for $\p{1}$ in the game $\game$.\\
		\[
		\safeactions{X_i}{v}{\forwardingactions{X_{i-1}}{v}{\subaction'}} = \safeactions{X_i}{v}{\actiont(v)}
		\]
		By the \cref{equ:defAset}, we know that $\safeactions{X_i}{v}{\actiont(v)} = \actiono(v)$. Therefore, it is easy to see that:
		\[
		\subaction' \subseteq \safeactions{X_i}{v}{\forwardingactions{X_{i-1}}{v}{\subaction'}}
		\]
		Finally, by using the \cref{lem:cobuchi1}, we can conclude that $\strat'$ is a winning strategy for $\p{1}$ in the game $\game$.\\
	\end{proof}
	Now, we are going to use the above lemmas to show that for every $v \in X_i \cap \targetSet$, $\strat \in \templatelive(\funcLive, \livePartitions)$.\\
	If $\strat$ is a winning strategy for $\p{1}$ in the game $\game$, for every $v \in X_i \cap \targetSet$, we compute $\subaction = \infStrat{\strat(v)}$. Now, we have two possibilities for $\forwardingactions{X_{i-1}}{v}{\subaction}$: \RomanNumeralCaps{1}. $\forwardingactions{X_{i-1}}{v}{\subaction} = \actiont(v)$ \RomanNumeralCaps{2}. $\forwardingactions{X_{i-1}}{v}{\subaction} \ne \actiont(v)$\\
	For case \RomanNumeralCaps{1}, by using \cref{lem:cobuchi2}, we can conclude that $\strat \in \templatelive(\funcLive, \livePartitions)$.\\
	For case \RomanNumeralCaps{2}, by using \cref{lem:cobuchi3}, we can find another winning strategy $\strat'$ for $\p{1}$ in the game $\game$, such that $\forwardingactions{X_{i-1}}{v}{\subaction'} = \actiont(v)$, where $\subaction' = \infStrat{\strat'(v)}$. Then, by using \cref{lem:cobuchi2}, we can conclude that $\strat' \in \templatelive(\funcLive, \livePartitions)$.\\
	Now, we showed that for every $v \in X_i \cap \targetSet$, every winning strategy $\strat$ either follows our template or we can change it to a new winning strategy $\strat'$, which will follow our template.\\
	Finally, since we showed that for every $v \in X_1$, $\strat \in \templatecolive(\funcCoLive)$, and for every $v \in X_i$, either $\strat$ or a new winning strategy is in $\templatelive(\funcLive, \livePartitions)$, and for every $v \in \wino$, $\strat$ is in $\templatesafe(\funcSafe)$, we can conclude that the template $\template = \templatesafe(\funcSafe) \cap \templatelive(\funcLive, \livePartitions) \cap \templatecolive(\funcCoLive)$ is non-empty.
\end{proof}

\section{Non-maximally permissiveness examples for templates}\label{app:non_maximally}

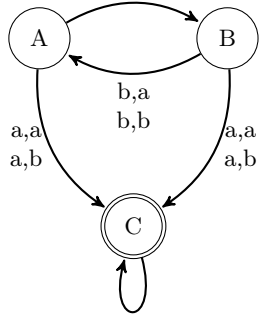
\begin{figure}[H]
\centering
\begin{tikzpicture}[>=stealth', auto, node distance=2.5cm, every loop/.style={min distance=10mm}]

  \node[state] (A) {A};
  \node[state, right of=A] (B) {B};
  \node[state, below of=A, xshift=1.25cm, accepting, double distance=1pt] (C) {C};

  \path[->]
    (A) edge[bend right] node[left] {\shortstack{a,a \\ a,b}} (C)
    (C) edge[loop below] node {} (C)
    (B) edge[bend left] node[below] {\shortstack{b,a \\ b,b}} (A)
    (B) edge[bend left] node[right] {\shortstack{a,a \\ a,b}} (C)
    (A) edge[bend left] node[below] {} (B);

\end{tikzpicture}
\caption{A Büchi game, where \( \targetSet = \{C\} \) is the goal set.}
\label{fig:buchiMaximality}
\end{figure}

	\begin{example}
		Consider the concurrent \buchi game $\game=(\gamegraph,\Box \Diamond \targetSet)$, where $\gamegraph$ is shown in \cref{fig:buchiMaximality} and $\targetSet=\{C\}$. The set of actions for both players at every state is $\actiono(v) = \actiont(v) = \{a,b\}$.
		$\BuchiAlgo{\targetSet}$ will compute the sets $X_0, X_1, X_2$ as follows:
		\begin{itemize}
			\item $X_0 = \emptyset$
			\item $X_1 = \{C\}$
			\item $X_2 = \{A,B,C\}$
		\end{itemize}
		The winning region for $\p{1}$ is $\wino = X_2$, and the set of states that can reach $C$ in exactly one step is $\Apre{1}{\wino}{\targetSet} = \livePartition_1 = \{A, B\}$.
		In states $A$ and $B$, $\p{1}$ can simply play $"a"$ to win the game. However, suppose $\p{1}$ instead chooses $"b"$ at state $A$ with probability $1$, and at state $B$ plays $"a"$ with probability $2^{-n}$, when $B$ is visited for the $n$-th time. This is still a winning strategy, because the probability of not reaching $C$ is $0$.
		In this case, we have a winning strategy where, for neither of the states in $\livePartition_1$, the sum $$\sum\limits_{\{i \mid \play_i=v\}} \minsupport{\strato(\play_{<i})}{\funcLive(v)}$$ is equal to $\infty$. Therefore, this strategy does not follow the template $\templatelive(\funcLive, \livePartitions)$ and the template $\template = \templatesafe(\funcSafe) \cap \templatelive(\funcLive, \livePartitions)$ is not maximal.
	\end{example}

	\begin{figure}[h]
		\centering
		
		\begin{tikzpicture}[shorten >=1pt, node distance=4cm, on grid, auto]
			
			\node[state, label=left:$\targetSet$]             (S0) {$S_0$};
			\node[state, label=right:$\targetSet$, right=of S0](S1) {$S_1$};
			\node[state, label=left:$\targetSet$, below=of S0](S2) {$S_2$};
			\node[state, label=right:$\targetSet$, right=of S2](S3) {$S_3$};
			\node[state, label=below:$\neg \targetSet$, below=of S2](S4) {$S_4$};
			
			\path[->]
			(S0) edge[loop above] node {} ()
			(S1) edge[loop above] node {} ()
			(S2) edge[] node {$(a,d)\; (b,e)$} (S0)
			(S2) edge[bend right=15] node {$(y,d)\;(x,f)$} (S1)
			(S2) edge node {$(y,e)\; (x,e)\; (a,f)\; (b,f)$} (S3)
			(S2) edge[bend left] node {$(y,f)\; (x,d)\; (a,e)\; (b,d)$} (S4)
			(S4) edge node {} (S2)
			(S3) edge[bend left] node {} (S2);
			
		\end{tikzpicture}
		\caption{Co-B\"uchi game $\game_2$ with winning condition $\LTLeventually\LTLalways \targetSet$, where $\targetSet=\{S_0,S_1,S_2,S_3\}$}
		\label{fig:cobuchiNotMaximal}
		
	\end{figure}
	\begin{example}
		Consider the concurrent \cobuchi game $\game=(\gamegraph,\LTLeventually \LTLalways \targetSet)$, where $\gamegraph$ is shown in \cref{fig:cobuchiNotMaximal} and $\targetSet=\{S_0, S_1, S_2, S_3\}$.
		$\CobuchiAlgo{\targetSet}$ will compute the sets $X_0, X_1, X_2$ as follows:
		\begin{itemize}
			\item $X_0 = \{S_0, S_1\}$
			\item $X_1 = \{S_0, S_1, S_2, S_3\}$
			\item $X_2 = \{S_0, S_1, S_2, S_3, S_4\}=\wino$
		\end{itemize}
		In this algorithm, $S_2$ has been added to $X_1$ because it is in $\AFpre{1}{\wino}{X_1}{X_0}$. (Because $\nu \subaction \cdot (\safeactions{\wino}{S_2}{\emptyset} \land \safeactions{X_2}{S_2}{(\forwardingactions{X_1}{S_2}{\emptyset})}$ is equal to $\{a,b,x,y\} \ne \emptyset$.)
		\cref{alg:cobuechi_temp} will compute $\funcLive$ for $S_2$ as below:
		\begin{itemize}
			\item $d \rightarrow \{a,y\}$
			\item $e \rightarrow \{b\}$
			\item $f \rightarrow \{x\}$
		\end{itemize}
		Then, $\funcLive(S_2) = \{\{a,y\}, \{b\}, \{x\}\}$.\\
		In this game, $\p{1}$ can play $\{a,b\}$ with equal probability in $S_2$, and win the game. This strategy does not follow the template $\templatelive(\funcLive, \livePartitions)$ of the \cref{alg:cobuechi_temp}, but it is still a winning strategy, hence the template is not maximal.\\
	\end{example}

\section{Supplementary Material for the Experiments}\label{app:experiments}
\subsection{Conversion of Turn-based Games to Concurrent Games}\label{subsection:game-conversion}
Our conversion from turn-based to concurrent games follows these steps:
\begin{enumerate}
\item \textbf{Transition merging}: We keep \p{1} states for the concurrent games, and every two consecutive transitions in the turn-based games were merged into a single transition in the concurrent games. More specifically, for each pair of consecutive transitions $(u, a, v)$ and $(v, b, w)$ in the turn-based game, where $u$ is a \p{1} state and $v$ a \p{2} state, we create a single transition $(u, (a,b), w)$ in the concurrent game. Afterward, we eliminate \p{2} states. We add self-loops to state without outgoing transitions. Note that this is possible because we consider \emph{alternating} turn-based games.
\item \textbf{Winning conversion}: Our tool requires state-based winning conditions. Any state $v$ that has an incoming transition $(u, (a,b), v)$ for which at least one of the original transitions ($(u, a, v)$ or $(v, b, w)$) appears in the winning condition in the original game, is added to the objective set in the concurrent game.
\end{enumerate}

We remark that this procedure is not practical for large instances and is only applied to smaller games.
After excluding very large instances that could not be processed, we selected 171 instances for our experimental evaluation.
The list of considered (and converted) models can be found in \cite{anonymous_2025_17357029}.

\subsection{\textsc{ConSTel} vs \textsc{PeSTel}}
\begin{figure}[h!]
\centering
\includegraphics[scale=0.5]{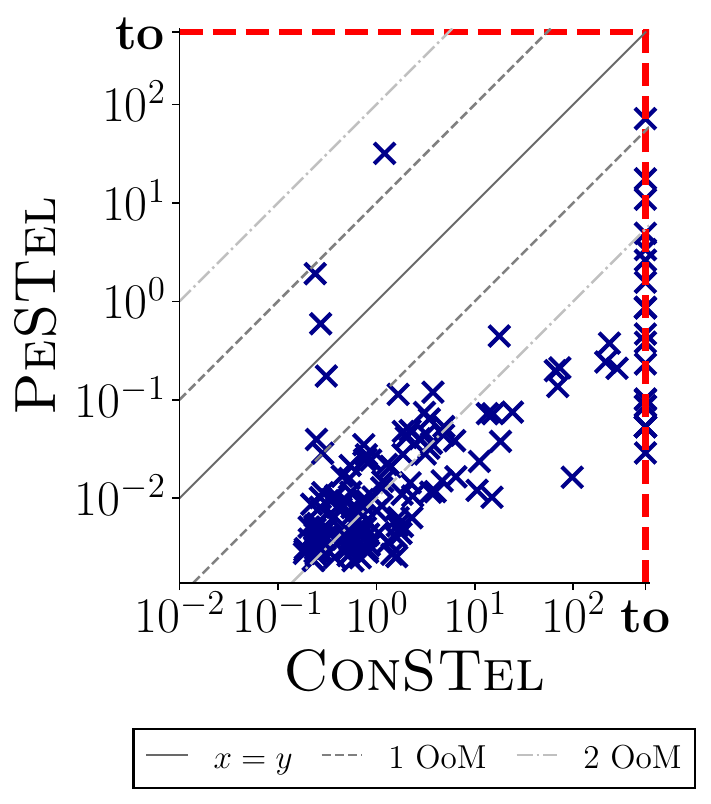}
\caption{Runtime comparsion}
\label{figure:runtime-comparison}
\end{figure}
We compared our prototypical tool \textsc{ConSTel} with the existing \textsc{PeSTel} tool \cite{AnandNS_PeSTels2023}, which computes strategy templates for \emph{turn-based} Parity games. However, since \textsc{PeSTel} can \emph{only} handle turn-based games, we considered the following setup. For each of the $171$ converted games, we ran the template computation in \textsc{ConSTel} and the template computation in \textsc{PeSTel} on the corresponding original turn-based game. This allows us to put the performance of \textsc{ConSTel} into perspective and empirically study the difficulty of going into the concurrent setting.
The results are summarised in the scatterplot in \cref{figure:runtime-comparison}.
Each point $(x, y)$ in the plot corresponds to one instance, where $x$ and $y$ are the runtime of concurrent \textsc{PeSTel} and \textsc{PeSTel}, respectively.
\textsc{PeSTel} often outperforms our tool by orders of magnitude (OoM), particularly for larger instances. This highlights the increased complexity inherent in solving concurrent games compared to turn-based games, and suggests that further engineering improvements are necessary to scale our tool to larger instances.

\subsection{Incremental Synthesis}
In \cref{sec:experiments} we have considered our strategy templates in the context of \emph{incremental synthesis}. Below we provide more details on the considered models:
\begin{itemize}
\item \texttt{full\_arbiter\_unreal1}: $55$ state and $731$ transitions
\item \texttt{lilydemo16}: $28$ state and $1228$ transitions
\item \texttt{full\_arbiter}: $55$ state and $731$ transitions
\item \texttt{amba\_decomposed\_tincr}: $26$ state and $289$ transitions
\end{itemize}

In \cref{figure:lineplot-conflict} we have aggregated the data over these selected models. In the following, we illustrate the conflicts for each of the games individually in \cref{fig:buchi_plus_buchi,fig:buchi_plus_cobuchi,fig:cobuchi_plus_cobuchi}. Each cell $(s, k)$ in the heatmaps corresponds to the percentage of conflicts that occurred when adding $k$ objectives of size $s$.

A general trend that can be observed is that the amount of conflicts increases the more objectives are added. This is expected since more templates need to be combined, thereby increasing the potential for conflicts. Another interesting observation is that conflicts seem to occur more often for medium-sized to larger-sized objective sets. A possible explanation is that for these sizes we get larger winning regions, again resulting in more potential for conflicts. Concerning the types of combined objectives, we observe that the combination of Büchi + Büchi appears to be least prone to conflicts. The introduction of new co-Büchi objectives results in more conflicts. Intuitively, this is the case because templates for co-Büchi objectives (compared to templates for Büchi objectives) also contain a co-live template.

\begin{figure}[h!]
  \centering
  
  \begin{subfigure}[b]{0.48\textwidth}
    \includegraphics[width=\textwidth]{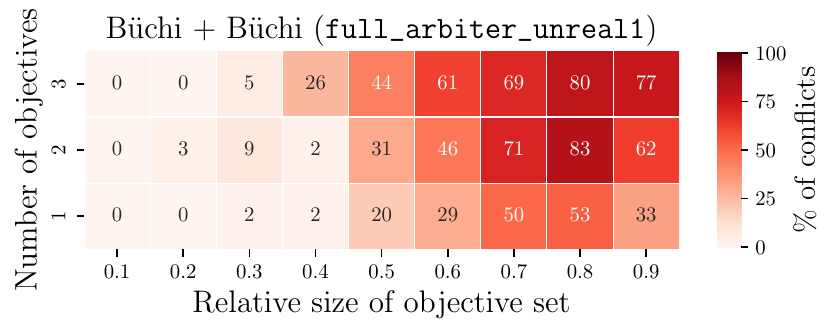}
    \caption{\texttt{full\_arbiter\_unreal1}}
  \end{subfigure}
  \hfill
  \begin{subfigure}[b]{0.48\textwidth}
    \includegraphics[width=\textwidth]{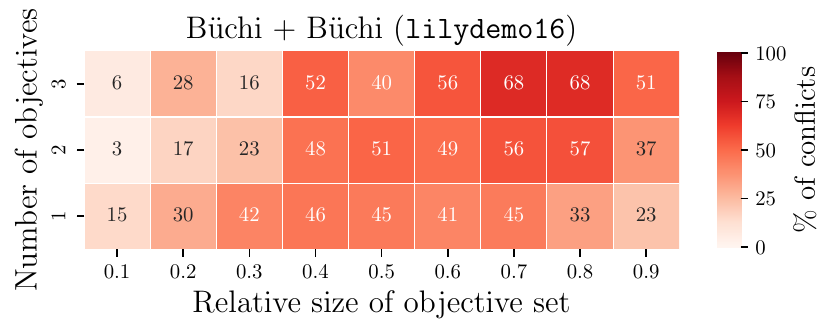}
    \caption{\texttt{lilydemo16}}
  \end{subfigure}
  
  \vspace{0.5em}
  \begin{subfigure}[b]{0.48\textwidth}
    \includegraphics[width=\textwidth]{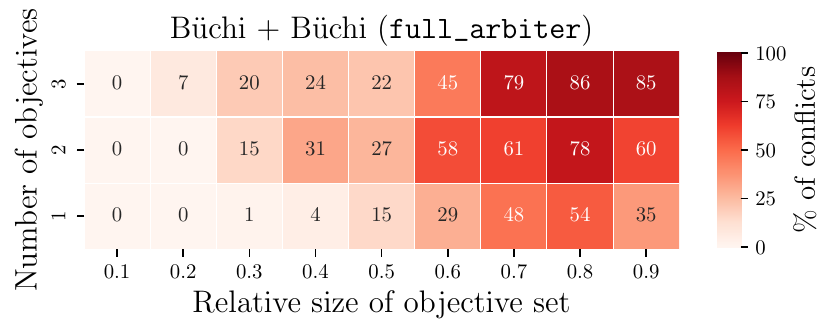}
    \caption{\texttt{full\_arbiter}}
  \end{subfigure}
  \hfill
  \begin{subfigure}[b]{0.48\textwidth}
    \includegraphics[width=\textwidth]{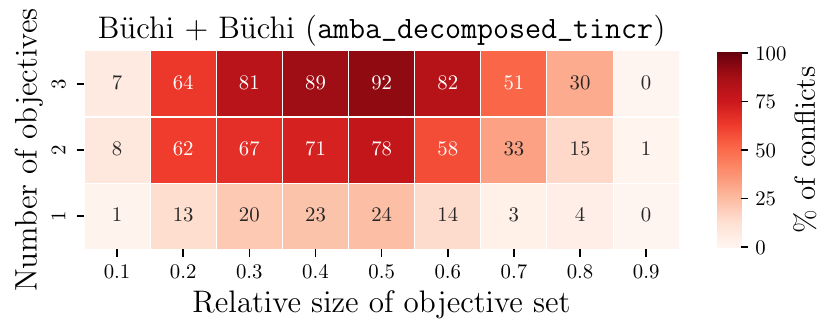}
    \caption{\texttt{amba\_decomposed\_tincr}}
  \end{subfigure}
  
  \caption{Conflict analysis for Büchi + Büchi each game.}
  \label{fig:buchi_plus_buchi}
\end{figure}

\begin{figure}[h!]
  \centering
  
  \begin{subfigure}[b]{0.48\textwidth}
    \includegraphics[width=\textwidth]{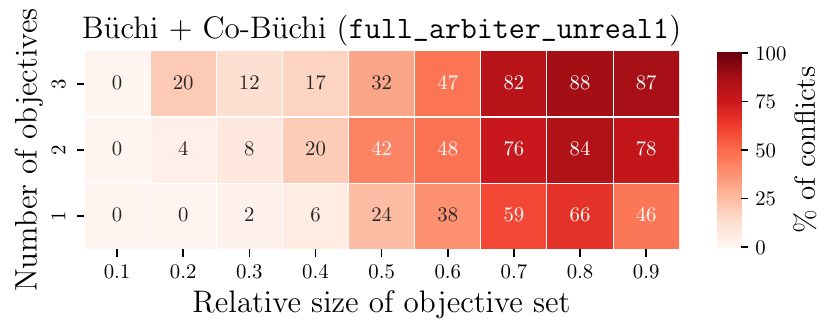}
    \caption{\texttt{full\_arbiter\_unreal1}}
  \end{subfigure}
  \hfill
  \begin{subfigure}[b]{0.48\textwidth}
    \includegraphics[width=\textwidth]{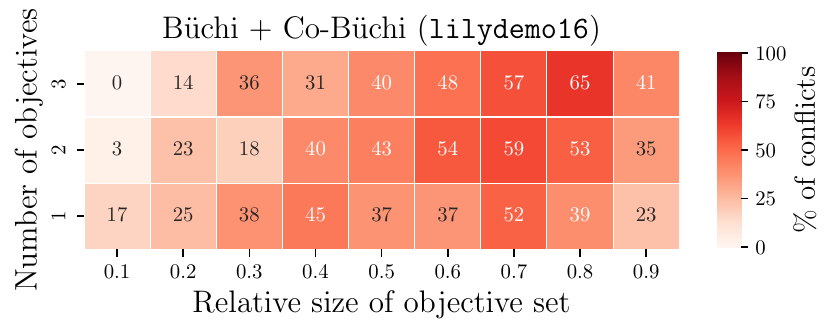}
    \caption{\texttt{lilydemo16}}
  \end{subfigure}
  
  \vspace{0.5em}
  \begin{subfigure}[b]{0.48\textwidth}
    \includegraphics[width=\textwidth]{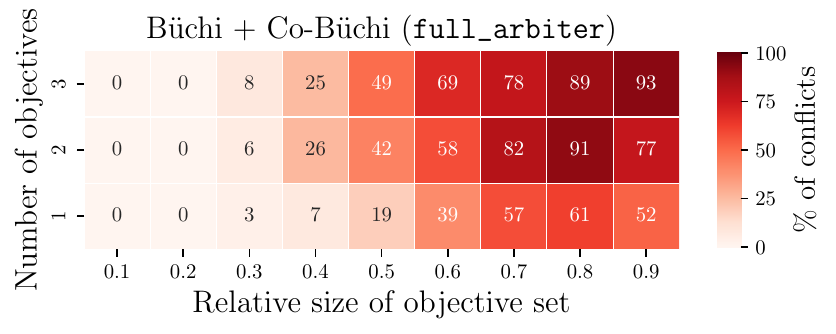}
    \caption{\texttt{full\_arbiter}}
  \end{subfigure}
  \hfill
  \begin{subfigure}[b]{0.48\textwidth}
    \includegraphics[width=\textwidth]{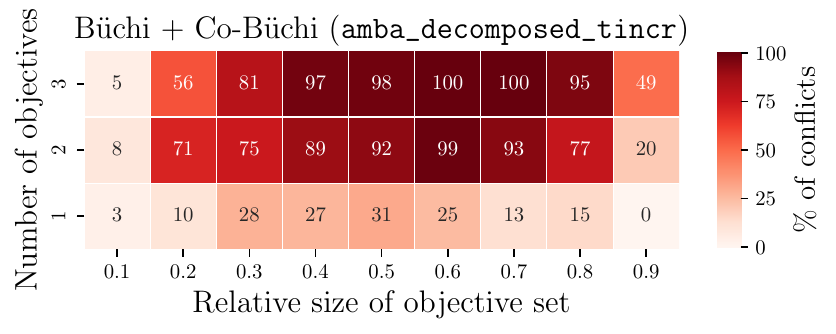}
    \caption{\texttt{amba\_decomposed\_tincr}}
  \end{subfigure}
  
  \caption{Conflict analysis for Büchi + Co-Büchi for each game.}
  \label{fig:buchi_plus_cobuchi}
\end{figure}

\begin{figure}[h!]
  \centering
  
  \begin{subfigure}[b]{0.49\textwidth}
    \includegraphics[width=\textwidth]{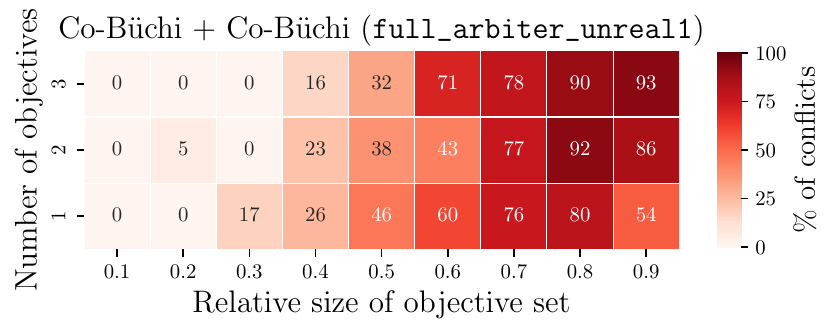}
    \caption{\texttt{full\_arbiter\_unreal1}}
  \end{subfigure}
  \hfill
  \begin{subfigure}[b]{0.49\textwidth}
    \includegraphics[width=\textwidth]{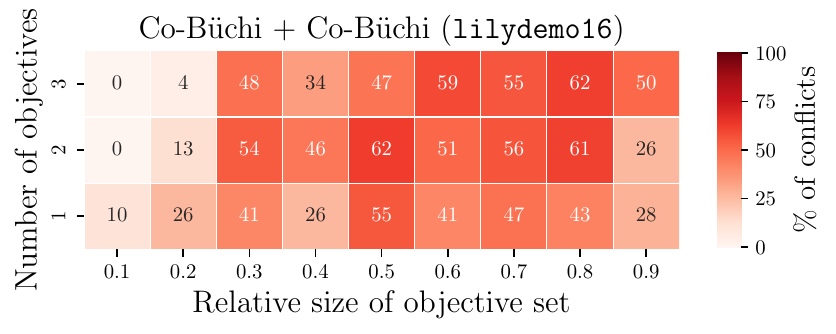}
    \caption{\texttt{lilydemo16}}
  \end{subfigure}
  
  \vspace{0.5em}
  \begin{subfigure}[b]{0.49\textwidth}
    \includegraphics[width=\textwidth]{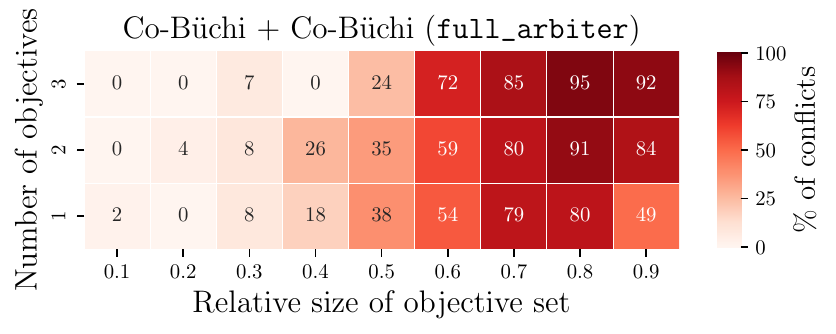}
    \caption{\texttt{full\_arbiter}}
  \end{subfigure}
  \hfill
  \begin{subfigure}[b]{0.49\textwidth}
    \includegraphics[width=\textwidth]{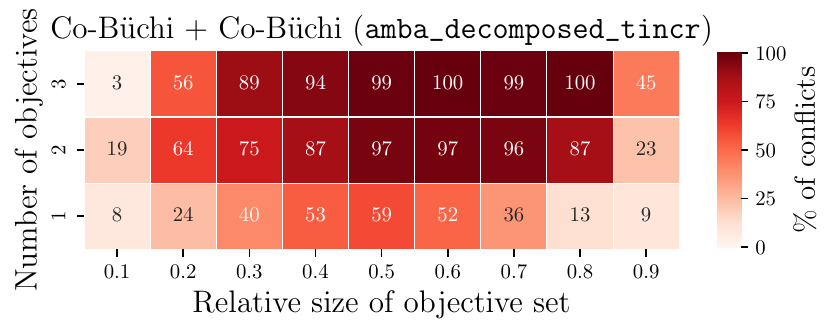}
    \caption{\texttt{amba\_decomposed\_tincr}}
  \end{subfigure}
  
  \caption{Conflict analysis for Co-Büchi + Co-Büchi for each game.}
  \label{fig:cobuchi_plus_cobuchi}
\end{figure}

\end{document}